\documentclass[10pt,journal]{IEEEtran}
\usepackage{cite}
\usepackage{amsmath,amssymb,amsfonts}
\usepackage{graphicx}
\usepackage{textcomp}
\usepackage{xcolor}
\usepackage{subfigure}
\usepackage{amsthm}
\usepackage{algorithmic}
\usepackage{algorithm}

\newtheorem{theorem}{Theorem}  [section]

\begin{document}
%
\title{Blockchain for Federated Learning: A Comprehensive Survey}



%
%
%

\title{Blockchain and Federated Edge Learning for Privacy-Preserving Mobile Crowdsensing
\author{Qin~Hu, Zhilin Wang, Minghui Xu, and Xiuzhen Cheng, \textit{Fellow, IEEE}}
\IEEEcompsocitemizethanks{
\IEEEcompsocthanksitem Qin Hu and Zhilin Wang are with the Department of Computer and Information Science, Indiana University-Purdue University Indianapolis, USA. \protect \\ E-mail: \{qinhu,wangzhil\}@iu.edu
\IEEEcompsocthanksitem Minghui Xu and Xiuzhen Cheng are with the School of Computer Science and Technology, Shandong University, China. \protect \\ E-mail: \{mhxu,xzcheng\}@sdu.edu.cn
}
}

\maketitle

\begin{abstract}
Mobile crowdsensing (MCS) counting on the mobility of massive workers helps the requestor accomplish various sensing tasks with more flexibility and lower cost. However, for the conventional MCS, the large consumption of communication resources for raw data transmission and high requirements on data storage and computing capability hinder potential requestors with limited resources from using MCS. To facilitate the widespread application of MCS, we propose a novel \textit{MCS learning framework} leveraging on blockchain technology and the new concept of edge intelligence based on federated learning (FL), which involves four major entities, including requestors, blockchain, edge servers and mobile devices as workers. 
Even though there exist several studies on blockchain-based MCS and blockchain-based FL, they cannot solve the essential challenges of MCS with respect to accommodating resource-constrained requestors or deal with the privacy concerns brought by the involvement of requestors and workers in the learning process. To fill the gaps, 
four main procedures, i.e., task publication, data sensing and submission, learning to return final results, and payment settlement and allocation, are designed to address major challenges brought by both internal and external threats, such as malicious edge servers and dishonest requestors. 
Specifically, a mechanism design based data submission rule is proposed to guarantee the data privacy of mobile devices being truthfully preserved at edge servers; consortium blockchain based FL is elaborated to secure the distributed learning process; and a cooperation-enforcing control strategy is devised to elicit full payment from the requestor.  
Extensive simulations are carried out to evaluate the performance of  our designed schemes.
\end{abstract}


%
\IEEEpeerreviewmaketitle

\maketitle

\section{Introduction}
%
%
%
%
\IEEEPARstart{F}{acilitated} by a variety of embedded sensors on mobile devices and ubiquitous Internet access opportunities, mobile crowdsensing (MCS) evolves into a vigorous paradigm \cite{nguyen2021federated}, benefiting the data collection of Internet of Things (IoT) and various practical applications, such as transportation monitoring \cite{farkas2015crowdsending}, localization and navigation  \cite{wang2018crowdnavi}. There are also an increasing number of mobile apps based on MCS, providing great convenience for our daily lives, such as Waze and Uber. In general, the success of MCS lies in the solicitation of distributed sensing capabilities of mobile devices across a large-scale area to serve for a major task in an environmentally friendly way, which can significantly reduce the cost of deploying traditional sensors in a large number of fixed locations.

Currently, mainstream research on MCS put efforts on incentive mechanisms design \cite{hu2019solving,wang2017efficient,wang2020towards}, quality control \cite{hu2019quality}, and privacy preservation \cite{zhao2020pace}, 
lacking attention on huge communication resource consumption for raw data transmission. Moreover, the large chunk of data collected from mobile workers put certain requirements on the data storage, processing and analysis capabilities of MCS requestors, which can prevent resource-constrained users from exploring and experiencing this promising computing model, thus hindering the wider application of MCS. 

To overcome these challenges, we propose a novel \textit{MCS learning framework} for the first time, leveraging on blockchain technology and the new concept of edge intelligence based on federated learning (FL). Four parts are involved in this framework in a top-down order, namely requestors, blockchain system, edge servers and mobile devices as workers, 
where the blockchain is maintained by a group of consortium members specified at the initialization of the MCS learning system while all other entities access to the blockchain via a client application. 

As a matter of fact, there exist plentiful studies about blockchain-based MCS \cite{wei2019blockchain,huang2020blockchain,hu2020blockchain,wei2020blockchain,kadadha2020sensechain,liang2020faircs,cai2019towards,an2018crowdsensing,an2020lightweight,chatzopoulos2018privacy,wang2018blockchain,gao2021trustworker,xu2021urim,zou2019crowdblps,yang2019blockchain,zhang2019vesenchain,wang2020blockchain,feng2019dynamic,noshad2019blockchain,gu2019using,jia2018blockchain,zhao2019dynamic}
and blockchain-based FL \cite{bao2019flchain,majeed2019flchain,kumar2021blockchain,qi2021privacy,kim2019blockchained,pokhrel2020decentralized,pokhrel2020federated,toyoda2019mechanism,weng2019deepchain,lu2019blockchain,kim2019blockchain,preuveneers2018chained,lu2020blockchain,kang2019incentive,awan2019poster,desai2021blockfla}.
Each can be further classified into tightly-coupled and loosely-coupled frameworks, where the tightly-coupled one employs blockchain to replace the originally centralized MCS platform or parameter aggregator in FL so as to achieve full decentralization and avoid the single point of failure, while the loosely-coupled framework implements the blockchain as an independent module to serve for some specific functions, such as the reputation management of MCS workers and FL participants. Our proposed framework compactly integrates blockchain into both MCS and FL, so the aforementioned existing schemes are not applicable to MCS learning due to the following two reasons. 
On the one hand, blockchain-based MCS still suffers from the huge cost of data transmission and storage, which becomes even worse by using the  blockchain since every blockchain node is supposed to hold a copy of the complete data about the main chain; besides, the requirement on data processing capability of the requestor continues to be a difficulty for resource-limited users. 
On the other hand, blockchain-based FL cannot be directly employed here, either, since the additionally involved requestors and mobile workers 
can bring new challenges to privacy protection in FL.

To fill the gaps, we design main procedures for our proposed MCS learning system, including i) task publication, ii) data sensing and submission, iii) learning to return final results, and iv)  payment settlement and allocation. The first procedure is initiated by the requestor sending an MCS learning request and then finished via a smart contract stored on the blockchain for worker recruitment. 
The second step is executed by the mobile devices and edge servers where devices undertake the sensing jobs according to the requirement of the requestor published on the blockchain and upload sensing data to the nearest edge servers for rewards. Edge servers can thus collect sensing data and establish local datasets for the next learning stage. FL is implemented in the third procedure where edge servers utilize their local data to collaboratively train the assigned machine learning (ML) model without leaking the raw data of mobile devices. And the final step is executed between the requestor and the blockchain once the task requirement is satisfied, where the blockchain system aims to successfully charge the payment from the requestor so that the fair distribution of rewards to all participants can be achieved for motivating their continuous contribution in the future.

However, there exist several challenges in the above-mentioned procedures. 
First, mobile devices 
may seriously concern whether the submitted sensing data to the edge servers can be protected from leakage; second, 
the conventional FL is vulnerable to malicious attacks on the centralized aggregator; third, 
the requestor might be reluctant to compensate the cost of all participants throughout the whole MCS learning process. Via solving these challenges, we make the following contributions:
\begin{itemize}
\item A novel and holistic MCS learning framework is proposed to conduct MCS and the subsequent data analysis in an integrated manner, which can not only lower the communication consumption and the 
requirement of computing and storage capabilities for requestors, but also preserve the privacy of workers at the local edge server level.
\item To guarantee that the submitted data of mobile devices are preserved at edge servers without leaking to others, we resort to mechanism design theory for devising an incentive-compatible game rule to restrict the privacy disclosing behaviors of edge servers.
\item For securing the learning process, consortium blockchain based FL is deployed, where the distributed ledger maintained by consortium nodes undertakes the coordination job and edge servers with local sensing data participate in the collaborative learning. 
\item To fully elicit compensation from the requestor, we employ the zero-determinant (ZD) game theory to work out a cooperation-enforcing control scheme, which is both theoretically and experimentally proved to be effective in driving requestors to pay fully, 
thus safeguarding the interests of all participants in the MCS learning system.
\end{itemize}


The remaining of this paper is organized as follows. We summarize the most related work in Section \ref{sec:relate} and present the overall system model in Section \ref{sec:system}. Detailed designs of main procedures are elaborated in Section \ref{sec:design}, followed by experimental evaluations on specifically designed schemes in Section \ref{sec:evaluation}. Finally, we conclude the whole paper in Section \ref{sec:conclusion}.

\section{Related Work}
\label{sec:relate}
Since the proposed MCS learning framework is largely unexplored, here we mainly summarize the most related work about blockchain-based MCS
\cite{wei2019blockchain,huang2020blockchain,hu2020blockchain,wei2020blockchain,kadadha2020sensechain,liang2020faircs,cai2019towards,an2018crowdsensing,an2020lightweight,chatzopoulos2018privacy,wang2018blockchain,gao2021trustworker,xu2021urim,zou2019crowdblps,yang2019blockchain,zhang2019vesenchain,wang2020blockchain,feng2019dynamic,noshad2019blockchain,gu2019using,jia2018blockchain,zhao2019dynamic}
and blockchain-based FL \cite{bao2019flchain,majeed2019flchain,kumar2021blockchain,qi2021privacy,kim2019blockchained,pokhrel2020decentralized,pokhrel2020federated,lu2019blockchain,kim2019blockchain,toyoda2019mechanism,weng2019deepchain,preuveneers2018chained,lu2020blockchain,kang2019incentive,awan2019poster
,desai2021blockfla}.

\subsection{Blockchain-based Mobile Crowdsensing}
Most of the research on blockchain based MCS deeply embed the blockchain system into MCS framework via employing the blockchain to undertake jobs of the traditional centralized MCS platform, such as task allocation, incentive mechanism implementation, and sensing data verification. In \cite{wei2019blockchain,huang2020blockchain,hu2020blockchain,wei2020blockchain,kadadha2020sensechain}, smart contracts were employed to handle the interactions between the requestors and workers in an automatic manner. Further, Liang \textit{et al.} \cite{liang2020faircs} used Trusted Execution Environments to avoid potential malice from requestors in blockchain-based MCS. 
To guarantee the data privacy of workers, Cai \textit{et al.} \cite{cai2019towards} employed secret sharing technique to enable flexible submission aggregation among blockchain nodes and workers. 
TrustWorker, a trustworthy and privacy-preserving worker selection mechanism for blockchain-assisted crowdsensing, adopted the deterministic encryption method to facilitate worker selection and privacy protection \cite{gao2021trustworker}. In \cite{xu2021urim}, URIM was introduced as an incentive mechanism to maximize the utilities of participants in blockchain-based crowdsensing. 
With more focus on specific performance improvement of MCS, An \textit{et al.} \cite{an2018crowdsensing,an2020lightweight} utilized the blockchain nodes as verifiers to achieve \textit{quality control}, and researchers in \cite{chatzopoulos2018privacy,wang2018blockchain,zou2019crowdblps,yang2019blockchain} took advantage of the distributed blockchain nodes to realize $k$-anonymity for workers so as to protect their \textit{location privacy} in MCS. Blockchain was also implemented in vehicular crowdsensing \cite{zhang2019vesenchain,wang2020blockchain} and radio frequency powered MCS \cite{feng2019dynamic}.

%

Other studies utilizing blockchain for MCS work in a loosely coupled fashion, where the blockchain network is used to facilitate some specific MCS procedure(s) as an independent function module rather than acting as the distributed MCS platform. In \cite{noshad2019blockchain,gu2019using}, blockchain was adopted to evaluate or verify the submissions from workers so as to eliminate their malicious behaviors and resist information tampering. Jia \textit{et al.} \cite{jia2018blockchain} designed a blockchain-powered confusion mechanism to hide the real locations of workers in MCS. While blockchain in  \cite{zhao2019dynamic} was mainly introduced to manage the reputation of workers at the network edge.

\subsection{Blockchain-based Federated Learning}

Blockchain applied in FL mainly aims at achieving fully distributed machine learning without a centralized aggregator. To coordinate learning procedures among all participated clients, blockchain stores all learning related information, such as initial model, local updates, and globally aggregated model. 
As blockchain can be classified into two general types, i.e., public and permissioned, researchers considered utilizing both blockchains in FL. For public chain based FL, FLChain was studied in \cite{bao2019flchain}.
Majeed \textit{et al.} \cite{majeed2019flchain} proposed another FLchain for mobile edge combing enabled FL. The work in \cite{kumar2021blockchain} utilized the public blockchain to verify the data for FL in the healthcare industry. Kim \textit{et al.} \cite{kim2019blockchained} designed BlockFL architecture and  analyzed the end-to-end latency. Pokhrel \textit{et al.} \cite{pokhrel2020decentralized,pokhrel2020federated} utilized the public blockchain with proof-of-work consensus algorithm to facilitate  on-vehicle machine learning.
And mechanism design was involved in \cite{toyoda2019mechanism} to guarantee the honesty of clients in public blockchain based FL platform. 
For permission chain based FL, Weng \textit{et al.} \cite{weng2019deepchain} devised DeepChain to achieve auditability of the whole FL training process with the help of Algorand consensus protocol. 
In \cite{qi2021privacy}, a permissioned blockchained FL system with the differential privacy method was designed to predict the traffic flow. Lu \textit{et al.} \cite{lu2019blockchain} integrated the FL in the consensus of permissioned blockchain and took advantage of it to realize privacy-preserving data sharing for industrial IoT. Dynamic weighting
scheme to improve learning performance was investigated in \cite{kim2019blockchain} for the private blockchain powered FL. And Preuveneers \textit{et al.} \cite{preuveneers2018chained} applied the permissioned blockchain based FL on intrusion detection. 
To adapt to the Internet of Vehicles (IoV) scenario with increasing security and reliability of FL, a hybrid blockchain system consisted of the permissioned chain and the local Directed Acyclic Graph
(DAG) was proposed in \cite{lu2020blockchain}, where a deep reinforcement learning based client selection was also designed to further improve the FL efficiency. Desai \textit{et al.} \cite{desai2021blockfla} designed a hybrid blockchain-based FL framework to automatically prevent the system from being attacked by malicious users via smart contracts.

The blockchain was also employed to operate independently to enhance FL. In \cite{kang2019incentive}, blockchain was employed to realize reputation management for FL clients in a reliable manner. And Awan \textit{et al.} \cite{awan2019poster} utilized the blockchain technique to guarantee provenance of model updates in FL so as to avoid intentional malice from clients.

It is clear that the existing studies are not applicable to MCS learning. 
In detail, blockchain-based MCS still fails to address main challenges of serving resource-limited requestors, while blockchain-based FL cannot be directly adopted since the involvement of requestors and mobile workers may hamper data privacy protection. In this paper, we design specific schemes based on mechanism design and game theory to solve these challenges, thus facilitating the function of the MCS learning system.


\section{System Model}
\label{sec:system}

\begin{figure}[htbp]
\centering
\includegraphics[width=0.49\textwidth]{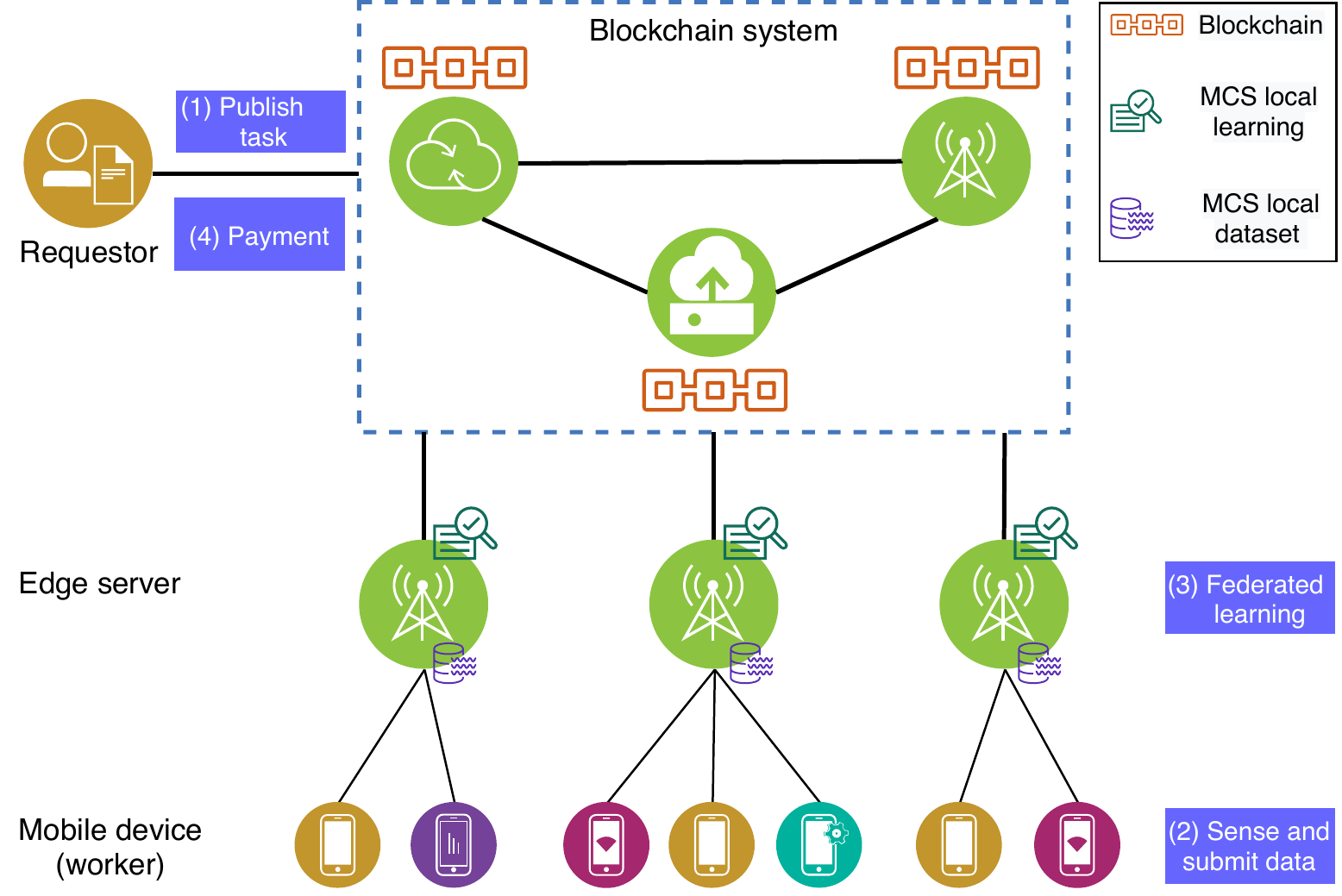}
\caption{The proposed MCS learning framework, consisting of requester, consortium blockchain, edge servers, and mobile devices. After the requester publishes the MCS task to the blockchain network, mobile devices will be recruited as workers to collect data and send them to the nearest edge servers, where the raw data stored at the edge servers will be proceeded via federated learning; appropriate rewards will be distributed to all participates once the learning process finishes.
}
\label{fig:framework}
\end{figure}

With the pervasive wireless network access opportunity, there is an increasing demand for mobile crowdsensing, where mobile devices (e.g., smart phones with embedded sensors) can collect sensing data from various locations all over the world as required by the requestor. On the one hand, the huge amount of sensing data can make the data processing and analysis a big challenge for potential MCS requestors with limited computation resource, hindering the large-scale application of MCS. On the other hand, nowadays, workers in MCS become increasingly concerned about the privacy leakage during the sensing data submission process, discouraging their active participation. To solve these challenges, we propose an MCS learning system with privacy preservation based on federated learning and blockchain, by which any requestor can easily obtain the MCS data analysis results online without worrying about the sequel data processing cost while mobile devices as workers can be assured to participate in data sensing with a certain degree of privacy protection.

As shown in Fig. \ref{fig:framework}, our proposed MCS learning system consists of four main parts, i.e., requestor, blockchain system, edge servers, and mobile devices as workers, where the blockchain is maintained by a predefined group of consortium members while all other entities connect to the blockchain system as clients via an interface. 
To accomplish an MCS learning request, we consider four general procedures as follows:
\begin{enumerate}
\item[(1)] \textbf{Task publication.} The requestor sends an MCS learning request to the blockchain network, which will be recorded on the main chain with necessary task information, such as the requestor's identity, task description, and performance requirement. This on-chain record will trigger a smart contract to recruit appropriate workers.

\item[(2)] \textbf{Data sensing and submission.} After workers receive the MCS task, they collect the required data in their convenience and submit sensing data to the nearest edge server so as to alleviate its own storage burden and accelerate the whole MCS process. All sensing data collected at the edge server constitutes a private dataset, which will be locally processed during the next stage. By this means, the private information of mobile devices (e.g., location) can be preserved to some extent as the raw data will only be visible to their local edge servers rather than the centralized platform in the conventional MCS. 
Besides, workers can receive sensing reward from edge servers in a timely manner. 

\item[(3)] \textbf{Learning to return final results.} To avoid privacy leakage, edge servers with local datasets will collaboratively conduct machine learning based on the federated learning (FL) framework to derive the required model for the requestor. 

\item[(4)] \textbf{Payment settlement and allocation.} Once the FL achieves the performance requirement, the requestor can easily obtain the final learning results from the blockchain system. To motivate continuous contribution from all participants during the whole MCS learning process, an appropriate payment will be charged from the requestor to cover sensing rewards for mobile devices, learning rewards for edge servers, and block rewards for blockchain nodes. 
\end{enumerate}

However, within the above working process, there exist several problems threatening the implementation of the MCS learning system. First, whether the edge server can keep the promise of protecting the data privacy for mobile devices remains an uncertainty; second, the centralized parameter aggregator in FL is prone to be attacked by both inside computation clients and outsiders, leading to the risk of inefficient learning or model breaches; third, given potential malicious requestors, it is challenging to successfully elicit enough payment from the requestor so that all contributors can be rewarded properly. All these concerns need to be carefully considered and well addressed in the detailed design of the MCS learning system, which will be discussed in the next section.

\section{Design of Main Procedures}\label{sec:design}
In this section, we detail the design of our proposed MCS learning system via specifying the aforementioned main steps and solving the aforementioned challenges. For reference, we summarize key notations used in this section in Table \ref{notation}. 

\begin{table}[htbp]
\centering
\caption{Key Notations.}
\begin{tabular}{|c|l|}
\hline
Notation   & Meaning                                                                                                                      \\ \hline
$\theta$    & The data privacy leakage degree of the edge server                                                                           \\ \hline
$v$        & \begin{tabular}[c]{@{}l@{}}The self evaluation of device on the value regarding\\  the sensing data\end{tabular}             \\ \hline
$s$        & The amount of profit paid by the server to the device                                                                        \\ \hline
$U_s$      & \begin{tabular}[c]{@{}l@{}}The expected utility of the edge server with respect to\\  the device's sensing data\end{tabular} \\ \hline
$R$        & \begin{tabular}[c]{@{}l@{}}The amount of reward obtained from the requestor\end{tabular}                \\ \hline
$r(\theta)$ & The extra reward of leaking the sensing data                                                                                 \\ \hline
$g(v,s)$   & \begin{tabular}[c]{@{}l@{}}The successful data collection probability between \\ the device and the edge server\end{tabular} \\ \hline
$U_d$      & The expected utility of the mobile device                                                                                    \\ \hline
$C$        & The overall cost of the device                                                                                               \\ \hline
$W$        & \begin{tabular}[c]{@{}l@{}}The total amount of incentive each requestor needs to pay\end{tabular}                          \\ \hline

$a_r$        & \begin{tabular}[c]{@{}l@{}}The action of the requestor finalizing the ML in blockchain\end{tabular}                          \\ \hline

$a_b$        & \begin{tabular}[c]{@{}l@{}}The action of the leader finalizing the ML in blockchain\end{tabular}                          \\ \hline

$\textbf{p}$        & \begin{tabular}[c]{@{}l@{}}The mixed strategies of the leader\end{tabular}                          \\ \hline

$\textbf{q}$        & \begin{tabular}[c]{@{}l@{}}The mixed strategies of the requestor\end{tabular}                          \\ \hline

\end{tabular}
\label{notation}
\end{table}


\subsection{Task Publication}\label{subsec:task}
Since there exists no centralized MCS platform in our proposed MCS learning system, the task publication step will be accomplished by the decentralized blockchain system. In detail, when the requestor needs to finish an MCS learning task, a request following the predefined format will be sent to the blockchain system as a transaction via the blockchain interface, which can uniquely characterize this request with  task-related information, such as the requestor's identity, task description, and learning performance requirement\footnote{Other information can also be defined to describe an MCS learning request but cannot be fully listed here. For example, the requestor may also submit a test dataset for accuracy evaluation, which can be stored on the blockchain with an address based on InterPlanetary File System (IPFS) \cite{benet2014ipfs}.}. 
Once entering the blockchain network, this request will be verified and widely forwarded to reach most of the consortium members so that it can be efficiently recorded on the main chain. 

After the request is visible on the main chain to all blockchain nodes, the smart contract for worker recruitment will be triggered, which is a piece of program stored on the blockchain with a specific address. To invoke this smart contract, the requestor sets the recipient of the request as the address of the smart contract and then the information included in this request will be processed as input of this worker recruitment program. Note that the logic of this smart contract can directly follow the ideas of existing worker selection and assignment algorithms for MCS \cite{wang2017efficient,wang2020towards}, where the only difference is that via using the smart contract, the worker recruitment process is executed by peering blockchain nodes in a decentralized manner. 

\subsection{Mechanism Design for Data Submission}\label{subsec:mechanism}
After the determination of worker recruitment in the previous stage, selected mobile devices can start collecting required sensing data at specific locations and then submit the sensed data to the nearest edge server for further processing. 
By doing so, devices can expect to achieve local privacy preservation at the level of edge servers as their submitted data will be processed and analyzed by the server locally, along with the collected sensing data from other connected mobile devices, instead of being directly submitted to a global platform. Nevertheless, even with this computation paradigm, devices may still have the concern about whether edge servers can really keep the local privacy protection expectation in practice without leaking the received sensing data to other parties, which is also a private behavior for the edge server and thus can never be explicitly known to devices. To solve this concern, we take advantage of the mechanism design theory \cite{hurwicz2006designing} which empowers devices to utilize the market power to restrain potential data privacy leakage behavior of edge servers. 

As indicated by \textit{revelation principle} in mechanism design, any Nash equilibrium of a Bayesian game (i.e., incomplete-information game) can be identically achieved in another incentive-compatible direct mechanism where game players report their private information truthfully. In the sensing data submission scenario, as mentioned above, we define the data privacy leakage degree of the edge server as its private information and denote it as $\theta \in [0,1]$, while the objective of the mobile device is to design a mechanism, or a game rule in other words, which can push the rational server to behave based on its real private information  $\theta$ for obtaining the optimum payoff. 

To maintain the long-term participation enthusiasm of workers in MCS, mobile devices should be compensated with sensing rewards. Since mobile devices accomplish sensing tasks with highly random trajectories, making it impractical for some devices to return to the previously connected edge servers for obtaining sensing rewards, we consider that edge servers pay devices for the submitted sensing data in an immediate manner. Specifically, the device has a self-evaluation on the value of the sensing data, denoted by $v$, to cover at least the sensing cost; and the server will decide how much profit it likes to pay to the device, denoted by $s$, which leads to a total payment of $v+s$ for the sensing data. 

During a period of time $[0,T]$, the expected utility of the edge server with respect to the device's sensing data can be defined as 
\begin{align}\label{eq:us}
U_s = \int_0^T  (R + r(\theta) - v - s) \cdot g(v,s) dt,
\end{align}
where $R$ is the amount of reward obtaining from processing it for the requestor, and $r(\theta) $ is the extra reward of leaking the sensing data, defined as $r(\theta) = \alpha_s \theta + \beta_s$ with coefficients $\alpha_s, \beta_s\geq 0$. Besides, $g(v,s)$ indicates the successful data collection probability between the device and the edge server, which is defined as
\begin{align}\label{eq:probability}
g(v,s) = \epsilon \frac{v}{\bar{v}} + (1-\epsilon) \frac{s}{\bar{s}}.
\end{align}
In the above equation, $\epsilon \in [0,1]$ is a scalar; $\bar{v}$ and $\bar{s}$ are the maximum values of $v$ and $s$, respectively. In detail, the successful probability is positively related to both the server's decision on the profit it willing to offer to the device and the device's self-assessment of data value. 
It is worth noting that the above definition of $g(v,s)$ is known by both sides prior to the mechanism design process.

At the same time, the expected utility of the mobile device can be calculated by
\begin{align}\label{eq:ud}
U_d = \int_0^T  ( v + s - C) \cdot g(v,s) dt,
\end{align}
where $C=c_d + \eta c_s$ is the overall cost of the device with $c_d$ denoting the sensing cost, $c_s$ representing the expected lost caused by the privacy leakage behavior of the server, and $\eta$ being a positive scalar. Considering that the privacy-leaking lost is related to the value of the device's sensing data, we define $c_s = \xi v$ where $\xi>0$.

According to mechanism design theory, $v$ and $s$ are the strategies of the device and the server, respectively, where $v$ is a function of $s$, acting as a game rule derived by the device for guaranteeing the privacy protection behavior of the server. In particular, the delicately designed game rule $v^*(s)$ can push the server to select the strategy $s$ based on its real private information $\theta$. The overall process is summarized in Algorithm \ref{al:2} and can be described as follows:
\begin{enumerate}
\item The device proposes a game rule $v^*(s)$ to maximize the expected utility $U_d$ and sends it to the server.
\item After receiving $v^*(s)$, the server calculates the best strategy $s^*$ maximizing the expected utility $U_s$ based on the hidden private information $\theta$. With $s^*$ and further derived $v^*(s)$, the server can determine whether it is profitable to accept the game rule or not. If the answer is yes, the server will send $s^*$ back to the device.
\item If the device receives the returned $s^*$ from the server within a certain time limit, $v^*$ can be calculated accordingly; otherwise, the device will terminate the sensing data submission process.
\end{enumerate}

Specifically, $v^*(s)$ and $s^*$ can be calculated by maximizing $U_d$ and $U_s$, respectively. To maximize $U_d$ in \eqref{eq:ud}, we denote the integrand as $F_d = ( v + s - C) \cdot g(v,s)  $.  Employing the calculus of variations method, we can derive $v^*(s)$ via solving the
associated Euler-Lagrange equation $\frac{\partial F_d}{\partial v} - \frac{d}{dt}\frac{\partial F_d}{\partial v'} = 0$ under the constraint $\frac{\partial^2 F_d}{\partial v^2}<0$. 
Accordingly, we can obtain that under the condition of $\eta \xi >1$, the optimal game rule of the device is
\begin{align}\label{eq:v_star}
v^*(s) = \frac{(1-\epsilon)(\eta \xi -1)\bar{v}s - \epsilon \bar{s}(s-c_d)}{2\epsilon(1-\eta \xi)\bar{s}}.
\end{align}
 
With the above $v^*(s)$, we can obtain $s^*$ through maximizing $U_s$ in \eqref{eq:us} using the similar method. To avoid redundancy, we omit the detailed calculation process of $s^*$ and report the result as follows. 
Let 
\begin{align*}
A_0 & = \frac{\epsilon}{\bar{v}}+\frac{(\eta \xi -1)(\epsilon-1)}{\bar{s}}, \\
A_1 & = 2(\eta \xi -1),
\end{align*}
 then we can have the optimal strategy of the edge server as
\begin{align*}
s^* = \frac{R+\alpha \theta + \beta +\frac{c_d}{A_1}}{2(\frac{\bar{v}A_0}{\epsilon A_1}+1)} + \frac{c_d \epsilon}{2\bar{v} A_1 (\frac{A_0}{A_1}-\frac{\epsilon-1}{\bar{s}})}.
\end{align*}

In fact, our proposed mechanism satisfies the incentive compatibility constraint, which can be demonstrated by the following theorem.

\begin{theorem}\label{thrm:incentive_compatible}
The game rule $v^*(s)$ proposed by the device is incentive-compatible.
\end{theorem}
\begin{proof}
Let $\hat{\theta}$ be the fake private information of the server. Then the optimal strategy of the server based on this $\hat{\theta}$, denoted as $\hat{s}^*$, will be sent to the device. As $\hat{\theta}$ is different from the real $\theta$, we have $\hat{s}^*\neq s^*(\theta)$. On the other hand, since $s^*(\theta)$ is calculated via maximizing $U_s$, we have $U_s(\hat{s}^*) < U_s(s^*)$, which makes the server's behavior of submitting the optimal strategy based on its fake private information not beneficial. Thus, any profit-driven and reasonable server will choose to respond with the real private information, which demonstrates the incentive compatibility of our proposed mechanism.
\end{proof}

Since the payment of $v+s$ for each device is covered in advance by the edge server, the total amount of payment for all locally connected devices accomplishing a specific MCS learning task should be reported to the blockchain system during the following FL step so that the edge server can be appropriately compensated finally from the requestor's payment. 
The specific allocation of compensation will be elaborated in Section \ref{subsec:payment}. 

\begin{algorithm}[t]
\caption{Mechanism Design based Data Submission} 
\label{al:2}

\begin{algorithmic}[1]

\STATE Raw data $\leftarrow$ Mobile devices collect data

\STATE $v^*(s)$ $\leftarrow$ The device determines the optimal game rule according to (\ref{eq:v_star})
\STATE The device sends $v^*(s)$ to the edge server
\STATE $U_s$ $\leftarrow$ The edge server calculates $U_s$  based on $v^*(s)$
\STATE $s^*$ $\leftarrow$ $\max (U_s)$ 
\IF{Accepting $v^*(s)$ is profitable}
\STATE The edge server sends $s^*$ to the local device
\STATE $v^*$ $\leftarrow$ $\max (U_d)$ 
\STATE The device submits the raw data to the edge server
\ELSE
\STATE The edge server keeps silent
\STATE The device terminates the sensing data submission process
\ENDIF
\end{algorithmic}
\end{algorithm}

\subsection{Blockchain-based Federated Learning at Edge}\label{subsec:blockchain}
With the received sensing data from mobile devices, edge servers can construct their local datasets so that they can further collaboratively conduct FL to derive the global learning results for requestors. However, there always exists an aggregator, such as a central server, to coordinate the FL process among all participated computing parties. Although some research suggest multiple servers to take the work of aggregation, the logical topology is still centralized, suffering from lots of weaknesses, such as single point of failure. 

To fundamentally address these problems, 
the concept of fully decentralized FL is proposed where FL participants communicate with each other via the peer-to-peer channels. As an exemplary implementation of this concept, blockchain \cite{lu2019blockchain,majeed2019flchain} based FL is widely studied recently as discussed in Section \ref{sec:related}. Thus, in our proposed MCS learning framework, we also consider to employ the decentralized ledger to replace the centralized aggregator in FL, where a set of blockchain nodes undertake the coordination jobs.
Specifically, the consortium blockchain is used here, which is maintained by a group of consortium members responsible for validating all content waiting to be included on the main chain as permanent records. 

For the reason of using a consortium blockchain instead of a public or private one to achieve distributed FL in our proposed MCS learning scenario, it is quite straightforward. On the one hand, either the FL model updates or MCS task related records on the blockchain are not supposed to be unlimitedly available to anyone from the public; on the other hand, blockchain nodes are not authorized by one centralized entity but several predefined members, such as trusted edge servers, and they need to be general and scalable enough to achieve a certain level of decentralization of the blockchain system. 
It is worth mentioning that although the blockchain system is mainly introduced here for FL implementation, records on the blockchain are not limited to ML model parameters and updates but also including MCS task related information defined in Section \ref{subsec:task}. 
 
Further, with the consortium members maintaining the distributed ledger, there are a lot of consensus protocols applicable to the consortium blockchain, such as PBFT 
and HotStuff \cite{yin2019hotstuff}. 
In our proposed system, we adopt PBFT as the consensus protocol in the consortium blockchain since it can improve the efficiency of reaching consensus and reduce time consumption. PBFT is one of the classical consensus protocols in blockchain, and its working process in our proposed blockchain can be described as below: at the beginning, the leader creates a block containing transactions, e.g., the local model updates, and then the leader broadcasts the block to other nodes in the blockchain, i.e., followers; next, all followers conduct the requested service (i.e., verifying the received block) and send back replies to the leader; in the end, if the leader receives $f+1$ replies from followers with the same result, where $f$ represents the maximum number of faulty nodes allowed, it means consensus has been reached in the blockchain system.
For stability and sustainability of the blockchain system, we consider that all blockchain nodes will receive incentives based on the information recorded on the main chain. In particular, any node generating a valid block will be finally rewarded by the payment from the requestor. This can be achieved due to the following three-aspect reasons: 1) the identity of blockchain node who generates a valid block is clearly indicated on the block; 2) the information records included in the block body can be indexed by the requestor identity or a unique task number assigned when the request is recorded on the blockchain; and 3) the payment settlement with the requestor and the specific reward allocation will be guaranteed by a well-designed scheme which will be introduced in the next section. 

\begin{figure}
\centering
\includegraphics[width=0.38\textwidth]{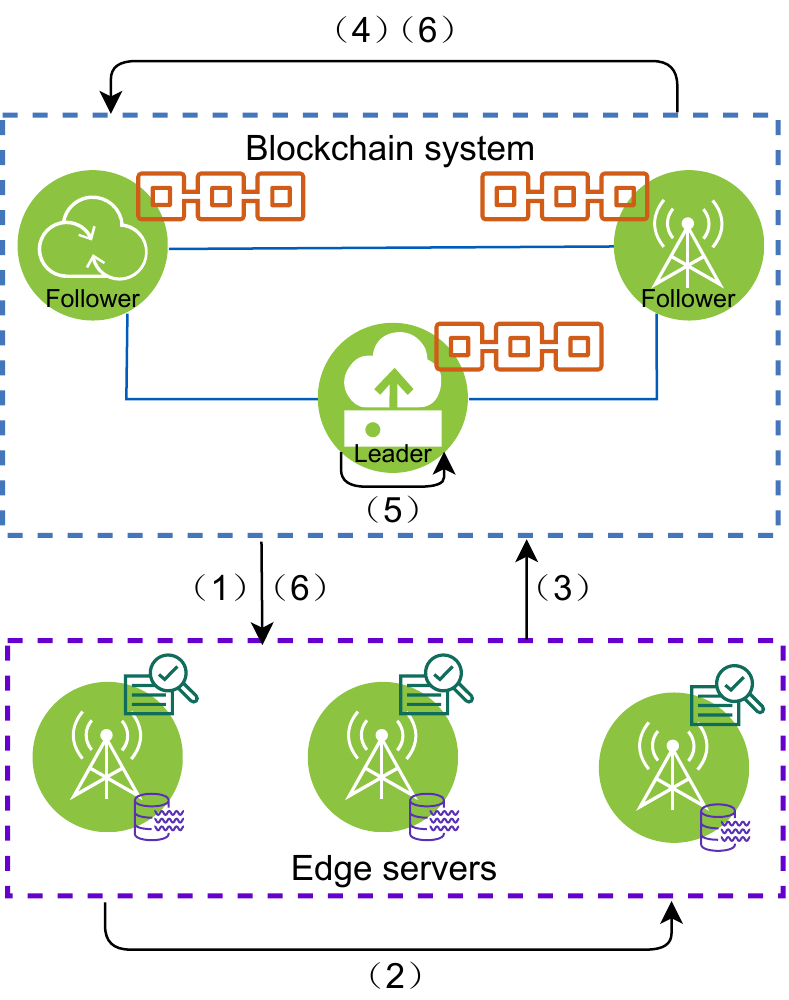}
\caption{The working process of blockchain-based FL, including six main steps: (1) initialization of ML model; (2) local training on edge servers; (3) submission of local model updates to the blockchain; (4) signature verification for the local model updates; (5) global model aggregation by the leader; (6) updating the aggregated model with edge servers.}
\label{fig:blockchain_FL}
\end{figure}

In Algorithm \ref{al:3}, we present the process of blockchain-based FL at edge.
To clearly illustrate the work flow of this process, we present it in Fig. \ref{fig:blockchain_FL} and describe main steps as follows:
\begin{enumerate}
\item[(1)] \textbf{Initialization.} If the requestor has a clear sense or requirement on the ML model, the initial model can be sent to the blockchain system during the task publication process as introduced in Section \ref{subsec:task}, which will be included on the blockchain for reference (Lines 1-2). While if the requestor has no idea on initial model selection, the blockchain system, especially the leader packaging the MCS learning request, can decide one via comparing the received task with the historical tasks\footnote{ 
Since the size of model parameters can be too large to fit in a block, we may employ the IPFS to store the model data, by which the records on the blockchain are referring to the address of the data. } (Lines 3-4).

\item[(2)] \textbf{Local learning.} Edge servers with local sensing datasets can obtain the ML model from the blockchain and then conduct the local training (Lines 6-8). 

\item[(3)] \textbf{Updates submission.} Once edge servers accomplish the local learning process, they submit model updates signed with their private keys to the blockchain with the corresponding task index for the convenience of model aggregation (Line 9).

\item[(4)]  \textbf{Signature verification.} After receiving local model updates from edge servers, blockchain nodes will verify the signatures to check whether the participants are legal or not (Line 11). If yes, these submissions will be broadcast to reach as many blockchain nodes as possible so as to be included in a valid block timely (Lines 12-13); otherwise, they will be discarded immediately (Lines 14-15). 

\item[(5)] \textbf{Model aggregation.} Due to the function of the consensus protocol, there will be one blockchain node, i.e., leader, in charge of generating a valid block each time, who is generated by Leader Election process \cite{castro1999practical} (Line 18). In this case, we consider that the leader is also responsible to calculate the aggregated model according to all submitted local model updates in the current round of FL as received in the blockchain system, which can avoid redundant or conflicting calculation to save computing resources, and then generate a new block including the aggregated model (Lines 19-20). 

\item[(6)] \textbf{Model updating.} If the global model is converged, the federated learning process can be terminated (Lines 21-22); otherwise, the above aggregated model recorded on the blockchain can be accessed by edge servers for updating their local models and begin the next round of FL (Lines 23-24). 


\end{enumerate}
The above process from (2) to (6) will be conducted repeatedly until the required model performance is realized, where the leader will explicitly associate the requestor identity with the finalized global model so as to be saved on the blockchain. 
Then the requestor can easily obtain the final result via the blockchain client.

\begin{algorithm}
\caption{Overall Process of the Proposed Blockchain-based FL at Edge} 
\label{al:3}
\begin{algorithmic}[1]
\REQUIRE Initial model $G(w)$, local model updates $f(w)$
\ENSURE Final model $G^*(w)$
\IF{Initial model $G(w)$ is specified by the requester}
\STATE $G(w)$ is recorded on the blockchain
\ELSE
\STATE $G(w)$ will be determined by comparing the received task with the historical tasks
\ENDIF
\WHILE{Not reaching the required model performance}
\STATE Edge servers obtain $G(w)$ from the blockchain
\STATE Local model updates $f(w)$ $\leftarrow$ Edge servers conduct local training to improve $G(w)$
\STATE Edge servers submit $f(w)$ with signature to the blockchain system
\WHILE{Blockchain node}
\STATE Verify the signature of $f(w)$
\IF{$f(w)$ is valid}
\STATE $f(w)$ will be broadcast to others
\ELSE
\STATE $f(w)$ will be discarded
\ENDIF
\ENDWHILE
\STATE Leader $i$ $\leftarrow$ Leader Election
\STATE New global model $G'(w)$ $\leftarrow$ Leader $i$ conducts model aggregation
\STATE New block $\leftarrow$ Block generation
\IF{$G'(w)$ is converged}
\STATE $G^*(w)$ $\leftarrow$ $G'(w)$
\ELSE
\STATE $G(w)$ $\leftarrow$ $G'(w)$
\ENDIF
\ENDWHILE
\RETURN $G^*(w)$\\
\end{algorithmic}
\end{algorithm}

\subsection{Payment Settlement and Allocation}\label{subsec:payment}
With the joint effort of mobile devices, edge servers, and blockchain nodes, the whole MCS learning system can function to meet the demand of the requestor. However, this system cannot work in the long term unless the incentives for all participants are guaranteed and allocated appropriately. In this section, we design a scheme for payment settlement and allocation, dealing with how to enforce requestors to pay for the accomplished MCS learning tasks. 

All the payment from the requestor is used to cover the incentives for data sensing of mobile devices, local learning of edge servers, FL model aggregation and block generation of blockchain nodes. In particular, the incentive for any mobile device has been immediately covered as $v+s$ by the edge server as introduced in Section \ref{subsec:mechanism}; while the incentives for FL local training, model aggregation, and block generation can be designed as flat rates for fairness consideration, denoted by $w_l$, $w_a$, and $w_b$, respectively. Note that all the above incentives can be calculated by tracking the recorded information on the blockchain, and thus the total amount of incentive each requestor needs to pay can be easily calculated out  by the leader who aggregates the final model, denoted by 
as $W$.

Here we consider that the requestor will use the MCS learning service repeatedly via interacting with the blockchain system. To provide better user experience to requestors, our proposed system updates the MCS learning result on blockchain first and then conducts the payment settlement and allocation. This sequential interaction process makes room for malicious requestors rejecting to fully pay the total amount of incentive $W$. To alleviate this problem, we propose to take advantage of zero-determinant (ZD) game theory \cite{hu2019quality} to safeguard the interests of all participants in the MCS learning system. 
Since the leader finalizing the global model works as the last step determining whether the requestor can successfully obtain the final MCS learning result, we consider to deploy the payment enforcement scheme there with a smart contract implementing the detailed strategy. Note that the smart contract is stored on the blockchain as a record so that any blockchian node being the leader can run this program if the finalized global model satisfies the requested model performance.

With $a_r$ and $a_b$ respectively denoting the actions of the requestor and the leader finalizing the ML model in blockchain, we define the requestor not paying the full amount of incentive $W$ as defection, denoted by $a_r=d$, and the action of paying $W$ as cooperation with $a_r=c$; while as the key point serving the requestor, the leader can opportunistically choose to update the final learning results to the main chain or not, denoted as $a_b =c$ and $d$, respectively. With each having two actions, there exist four combinations of game states, i.e., $a_b a_r = \{cc,cd, dc,dd\}$. 
And their payoffs under all cases can be expressed as $\mathbf{x} = (x_1,x_2,x_3,x_4)$ for the leader and $\mathbf{y} = (y_1,y_2,y_3,y_4)$ for the requestor. Referring to the practical outcomes of these four cases, we can observe that 1) the defector in $cd$ and $dc$ states can gain the highest payoff while the cooperator receives the lowest payoff; 
2) $cc$ leads to the second highest payoffs for both since they cost and gain normally; and 3) $dd$ brings the second lowest payoffs for both as nobody in this case loses anything. 
Thus, there exist the relationships $ x_3 > x_1 > x_4 > x_2$ and $y_2 > y_1 > y_4 > y_3$.

For the game being repeated, we define the mixed strategies of the leader and the requestor as $\mathbf{p}=(p_1,p_2,p_3,p_4)$ and $\mathbf{q}=(q_1,q_2)$, respectively, where $p_i,~i \in \{1,2,3,4\}$ is the probability of $a_b=c$ given each game result in the last round, while $q_i,~i \in \{1,2\}$ is the probability of $a_r=c$ when $a_b=c$ or $a_b=d$ in this round. 
Note that the difference between the definitions of $\mathbf{p}$ and $\mathbf{q}$ lies in their action order, where the leader performs according to the action result in the last round while the requestor behaves based on the leader's action in the current round.

According to the extended version of ZD strategy we previously proposed in \cite{hu2019quality}, we can see that it becomes beneficial for the leader to be the first mover if we empower the leader to use the ZD strategy. By this means, the leader can unilaterally control the relationship between the expected payoffs of both sides and thus enforce the desired game result. To be specific, when the strategy of the leader $\mathbf{p}$ satisfies  $\tilde{\mathbf{p}}=(p_1-1,p_2-1,p_3,p_4) = \alpha \mathbf{x} + \beta \mathbf{y} + \gamma \mathbf{1}$ with $\mathbf{1}$ denoting a vector of four ones, a linear relationship between their expected payoffs can be enforced as
\begin{align*}
\alpha E_b + \beta E_r +\gamma = 0,
\end{align*}
where $E_b$ and $E_r$ are respectively the expected payoffs of the leader and the requestor in the long term. 

In our case, as the leader desires to enforce the long-term cooperation from the requestor, mutual cooperation becomes a feasible and stable solution, where their expected payoffs are $x_1$ and $y_1$. Thus, inspired by the \textit{cooperation-enforcing control} in \cite{hao2018payoff}, the leader can set a specific linear relationship as $E_b-x_1 = \chi (E_r - y_1)$ with $\chi \geq 1$ to drive the full cooperation of the requestor, which yields the strategy of the leader $\mathbf{p}$ satisfying  $\tilde{\mathbf{p}}=(p_1-1,p_2-1,p_3,p_4) = \gamma ((\mathbf{x}-x_1 \mathbf{1}) - \chi (\mathbf{y} - y_1 \mathbf{1}))$ ($\gamma \neq 0$). The effectiveness of this scheme can be demonstrated in the following theorem.

\begin{theorem}
When the leader sets the strategy $\mathbf{p}$ to meet  $\tilde{\mathbf{p}}=\gamma ((\mathbf{x}-x_1 \mathbf{1}) - \chi (\mathbf{y} - y_1 \mathbf{1}))$ where $\gamma \neq 0$, the stable state is mutual cooperation.
\end{theorem}
\begin{proof}
As a utility-driven player, the requestor usually hopes to maximize the payoff. 
With the leader setting the strategy $\mathbf{p}$ to meet  $\tilde{\mathbf{p}}=\gamma ((\mathbf{x}-x_1 \mathbf{1}) - \chi (\mathbf{y} - y_1 \mathbf{1}))$, there exists $E_b-x_1 = \chi (E_r - y_1)$. In this case, if the requestor gains $E_r$ higher than $y_1$, leading to $E_b-x_1$ with a value of $\chi$ times $E_r-y_1$, both can gain payoffs higher than those at the mutual-cooperation state, i.e., $y_1$ and $x_1$, at the same time. This cannot be true since the requestor can obtain the payoff larger than $y_1$ only if the leader gets the payoff less than $x_1$, while the leader with the payoff larger than $x_1$ leads to the requestor obtaining the payoff less than $y_1$. 
Thus, the highest expected payoff that the requestor can obtain is $E_r= y_1$, which leads to $E_b = x_1$, corresponding to the mutual-cooperation state in the long term.
\end{proof}

With the function of the above ZD-based payment enforcement scheme, the total amount of incentives $W$ will be sent to the MCS learning system and recorded on the blockchain. All blockchain nodes and edge servers can fetch their respective rewards in a transparent and honest way according to the task accomplishing records on the blockchain, which will also be included on the main chain for potential error tracking.

\section{Experimental Evaluation}\label{sec:evaluation}
In this section, we conduct extensive experiments to evaluate our proposed MCS learning framework via investigating the performance  of specific schemes designed for main procedures. Since the first procedure in Section \ref{subsec:task} mainly introduces the input of the MCS learning system without algorithm design, we only present the experimental results of the procedures from Section \ref{subsec:mechanism} to Section \ref{subsec:payment}.

\subsection{Data Submission Mechanism Design}

As proved in Theorem \ref{thrm:incentive_compatible}, given the mechanism design for data submission process where the mobile device releases the game rule $v^*(s)$, the edge server can obtain the maximized utility only when the real data leakage degree $\theta^*$ is indicated to form the best response strategy $s^*$. Since the private $\theta^*$ of the server is not available for the device to observe, we cannot evaluate its impact. 
Instead, we study the impacts of three observable parameters, i.e., $\eta$, $\xi$ and $\epsilon$, on the maximized utilities of both sides.
In detail, we set $\eta,\xi\in [0,15]$, $\epsilon \in [0,1]$, and   $R=10,\alpha_s = 3, \beta_s = 2, \bar{v}=50,\bar{s}=500,c_d=2,\theta=0.5$. Note that other parameter settings are also evaluated, which present similar trends and thus are omitted here.

We first investigate the impacts of the device's cost parameters $\eta$ and $\xi$ when  $\epsilon $ is set as 0.9. 
As shown in Fig. \ref{fig:utility_eta}, with the increase of $\eta$ and $\xi$, the maximized utility of the device will decrease sharply to a stable value while that of the server presents different trends. With the increasing $\eta$ and $\xi$, the server's maximized utility increases from zero to the largest point firstly and then gradually decreases to a stable value. This is because the increase of cost parameters directly affect the utility of the device, where the higher the cost, the lower the utility; while for the server, the maximized utility $U_s$ is impacted indirectly through the optimal strategies $s^*$ and $v^*(s)$.

\begin{figure}[htbp]
\centering
\subfigure{
\includegraphics[width=0.23\textwidth]{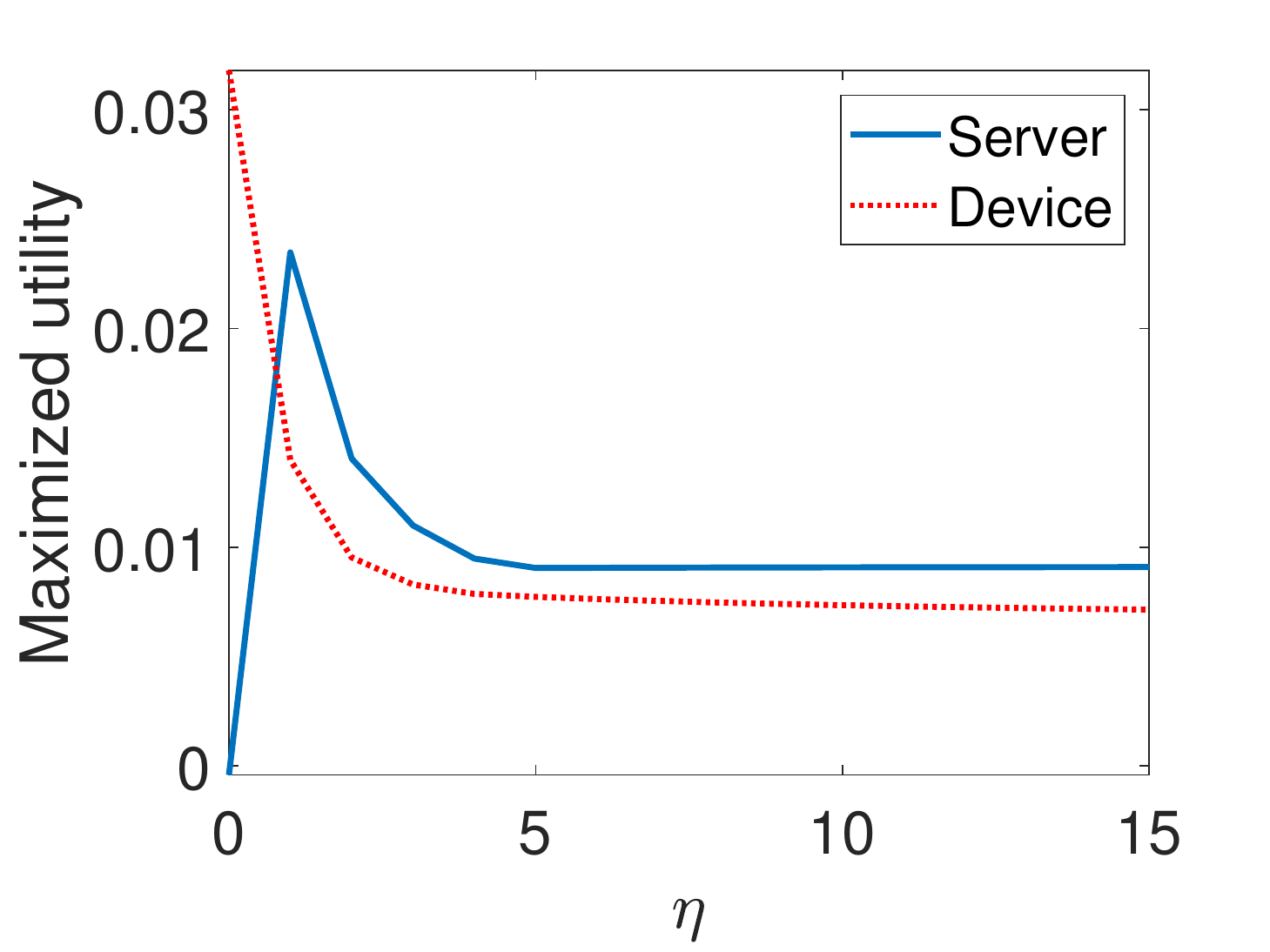}}
\subfigure{
\label{fig:ex4_alld}
\includegraphics[width=0.23\textwidth]{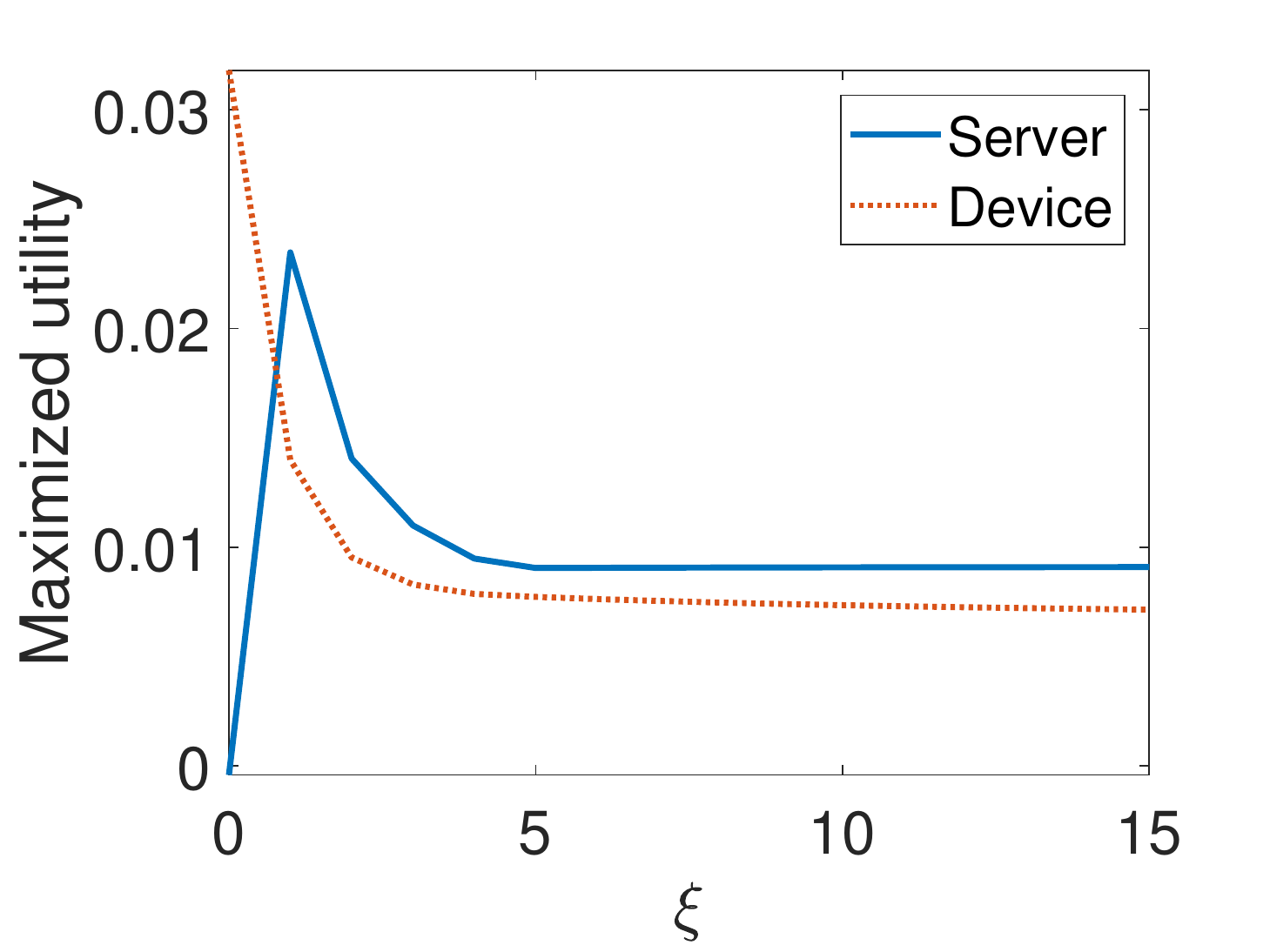}}
\caption{Impacts of $\eta$ and $\xi$ on the maximized utilities.}
\label{fig:utility_eta}
\end{figure}

Then we evaluate the impacts of $\epsilon$ which is the main parameter determining the successful data collection probability. As presented in Fig. \ref{fig:utility_epsilon}, the impact of $\epsilon$ on the server differs from that on the device. For the server, the maximized utility increases first and then decreases to almost zero; while the device's maximized utility decreases all the way to zero. According to \eqref{eq:probability}, the increase of $\epsilon$ implies that the optimal strategy of the device $v^*$ matters more for the successful probability. While referring to the expression of $v^*$ in \eqref{eq:v_star}, one can see that the larger the $\epsilon$, the lower the $v^*$, which leads to the decrease of $U_s$ according to \eqref{eq:ud}. But for the server, the increasing $\epsilon$ makes the decreasing $v^*$ result in a larger $U_d$ according to \eqref{eq:us} at the initial stage; when $v^*$ decreases to a certain degree, the decreased successful probability $g$ gradually functions to decrease $U_d$.

\begin{figure}[htbp]
\centering
\subfigure{
\includegraphics[width=0.23\textwidth]{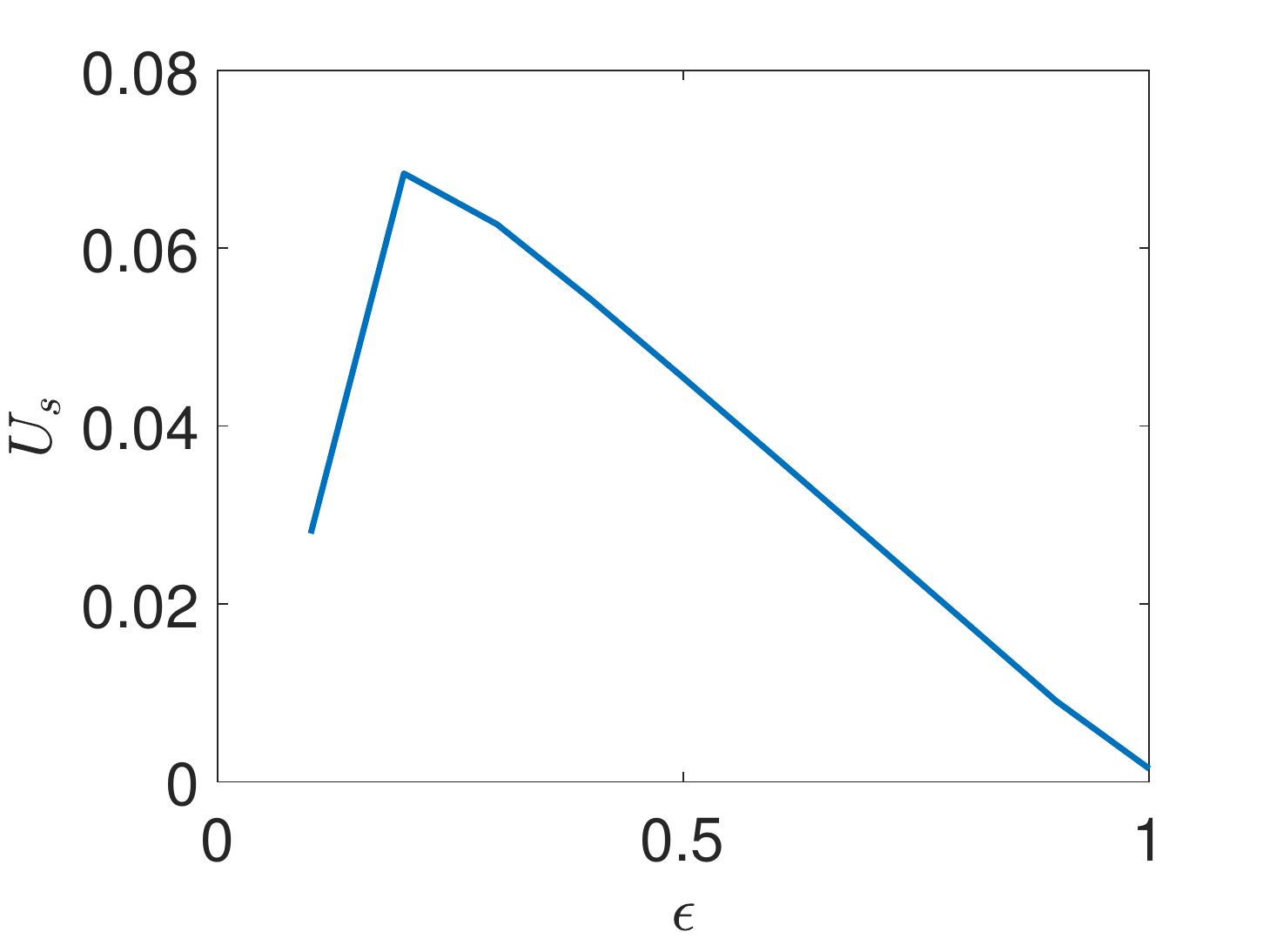}}
\subfigure{
\label{fig:ex4_alld}
\includegraphics[width=0.23\textwidth]{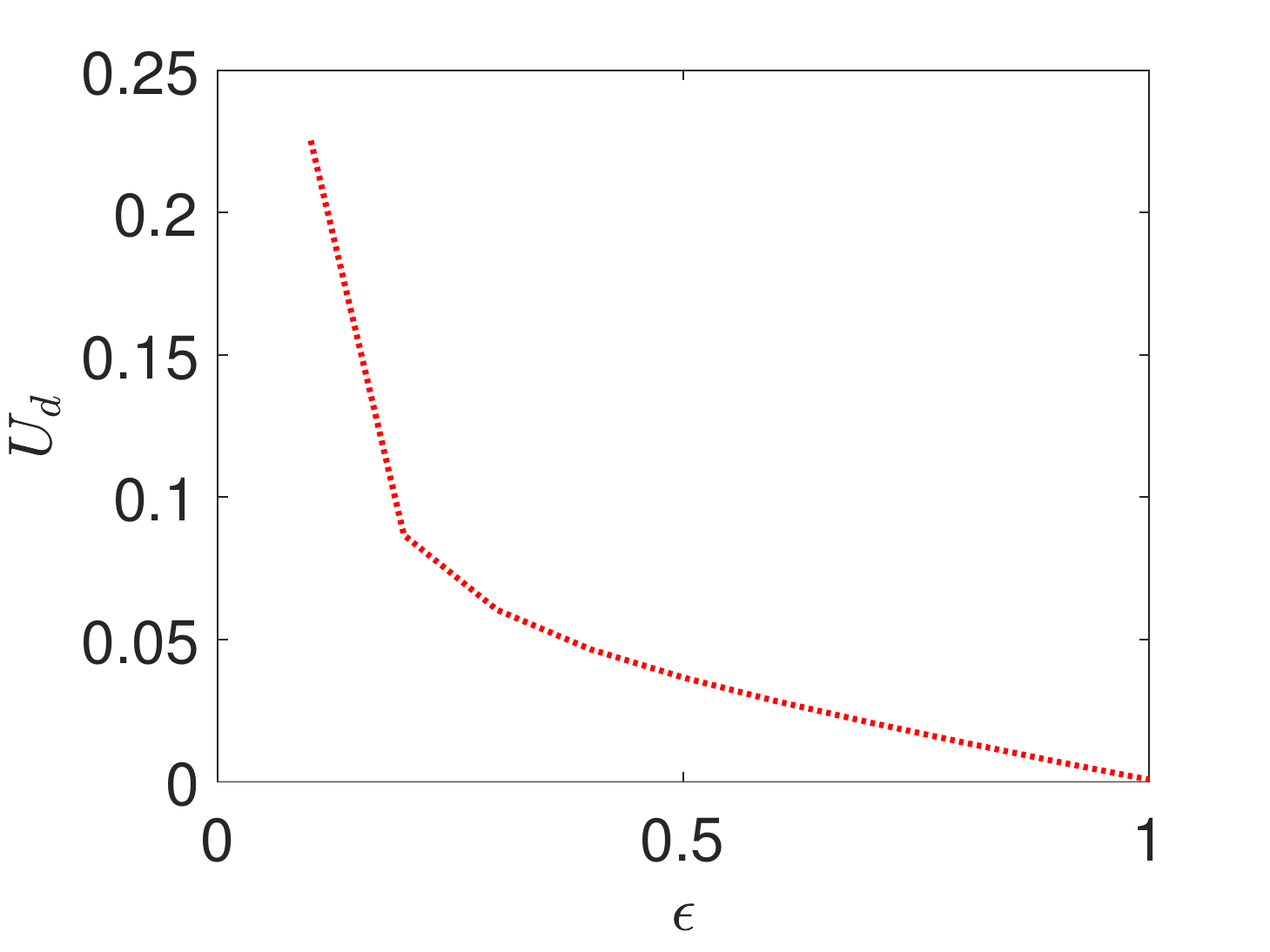}}
\caption{Impacts of $\epsilon$ on the maximized utilities.}
\label{fig:utility_epsilon}
\end{figure}

\subsection{Blockchain-based Federated Learning at Edge}
Considering that the consortium blockchain system for model parameter  coordination 
has limited impact on the FL performance, we conduct the simulation of FL first and focus on investigating the influences of available data size and edge servers' participation on learning results during the FL process. 
To simulate the FL process conducted among edge servers, we utilize a benchmark named LEAF \cite{caldas2018leaf} to execute federated learning using a 2-layer convolutional neural network (CNN) classifier for the FEMNIST dataset with 805,263 samples in total and 3,550 available participants for local computing. 
In our experiments, we first change the total data size used for FL indicated by the varying dataset fraction as 0.001, 0.005, and 0.05, with the fixed ratio between training and test dataset sizes as 9:1 and 35 FL participants; and then the number of participates is changed when the total amount of data used for FL keeps constant, where the participant fraction is set as 0.01, 0.03, and 0.05 to approximately simulate the number of participated edge servers as 35, 105, and 175, respectively.

The learning results under two parameter settings are reported in Fig. \ref{fig:fed_learning}. It can be seen from the left figure that given the same number of FL participants, the higher the dataset fraction used for FL, the better the learning performance with respect to the testing accuracy. 
With too few data samples (blue), the global model cannot even converge after 2,000 rounds of FL. 
While from the right figure, one can see with the same amount of data used for FL, the more the participants, the slower to reach convergence. This is reasonable since more participants splitting the total dataset will result in fewer data samples available for each one, leading to the longer training time. Corresponding to the MCS learning scenario, given the fixed amount of mobile devices to collect sensing data, the more edge servers involved, the smaller the local dataset of each edge server, and thus the slower the learning process.

\begin{figure}[htbp]
\centering
\subfigure{
\includegraphics[width=0.23\textwidth]{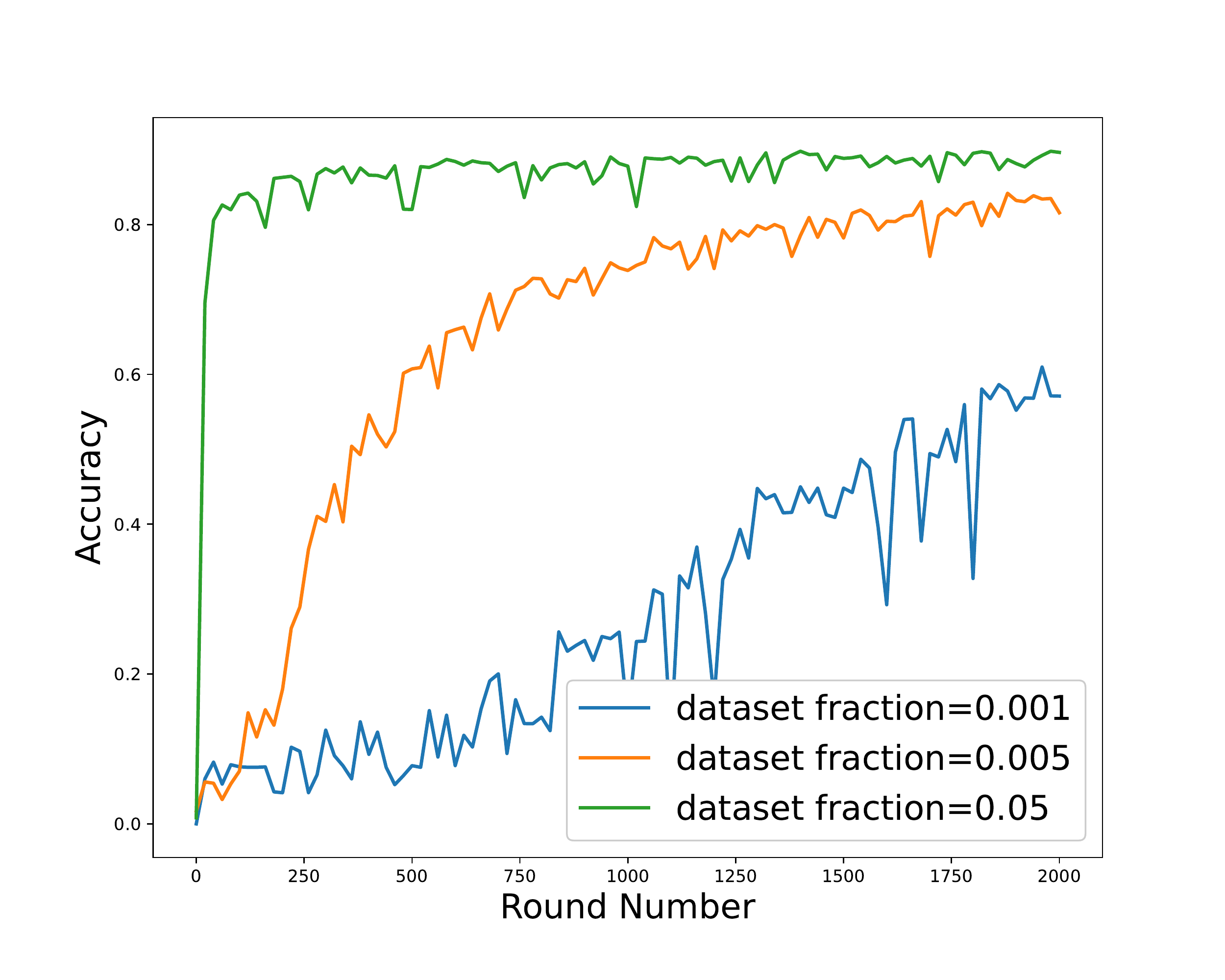}}
\subfigure{
\label{fig:ex4_alld}
\includegraphics[width=0.23\textwidth]{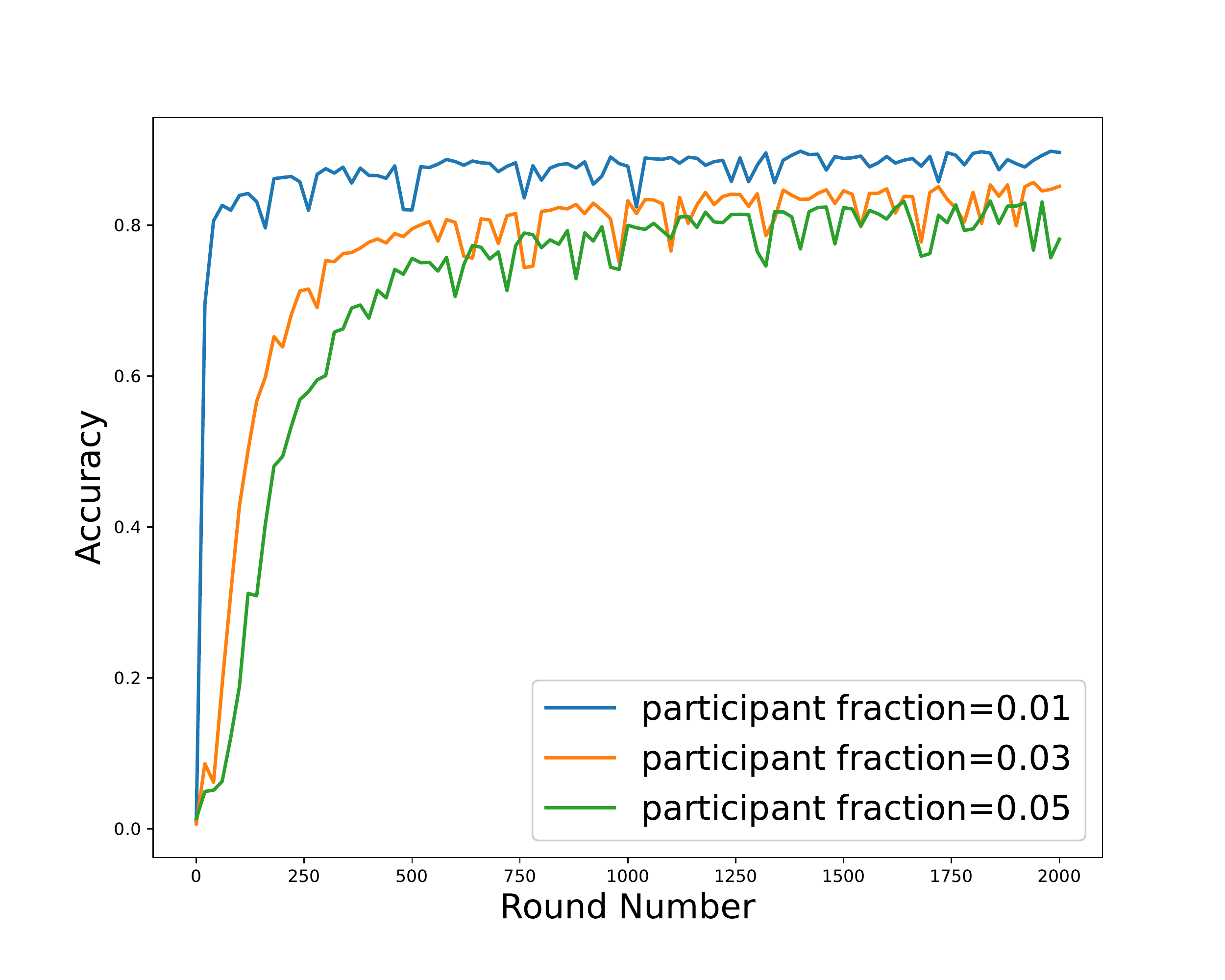}}
\caption{Impacts of data size and the number of edge servers on the learning accuracy.}
\label{fig:fed_learning}
\end{figure}

After the edge servers submit the local model updates to the blockchain, blockchain nodes generate a new block to include all updates and the aggregated model, which has to be supported by the predefined consensus protocol. To explore how the blockchain works in our framework, we design experiments to simulate the PBFT consensus process based on the framework in \cite{pbftconsensus} as an example. For convenience, we assume that each edge server submit the local model updates with a fixed data size of 300 KB. We implement this set of simulations using Python 3.8.5 in macOS 11.0.1 running on Intel i7 processor with 32 GB RAM and 1 TB SSD. 

First, we explore the relationship between the consensus efficiency in the blockchain system and the number of edge servers as FL clients given 10 nodes in the consortium blockchain. From Fig. \ref{fig:node}, we can see that the time cost to reach consensus is linearly correlated to the number of edge servers. This is because increasing the number of edge servers will lead a larger amount of data that need to be packed into blocks, resulting in more time to synchronizing the data on blockchain. 

Then we examine how the number of blockchain nodes affects the consensus efficiency when the number of FL participants (i.e., edge servers) is 10. The experimental results are shown in Fig. \ref{fig:node}.
It is clear that the time consumption increases as the number of blockchain nodes grows. 
The underline reason is that in the blockchain system running on a peer-to-peer networking structure, more nodes will lead to a larger increase in the communication times among them, and thus making the time consumed for reaching consensus increase near-exponentially.

\begin{figure}[H]
\centering
\subfigure{
\includegraphics[width=0.23\textwidth]{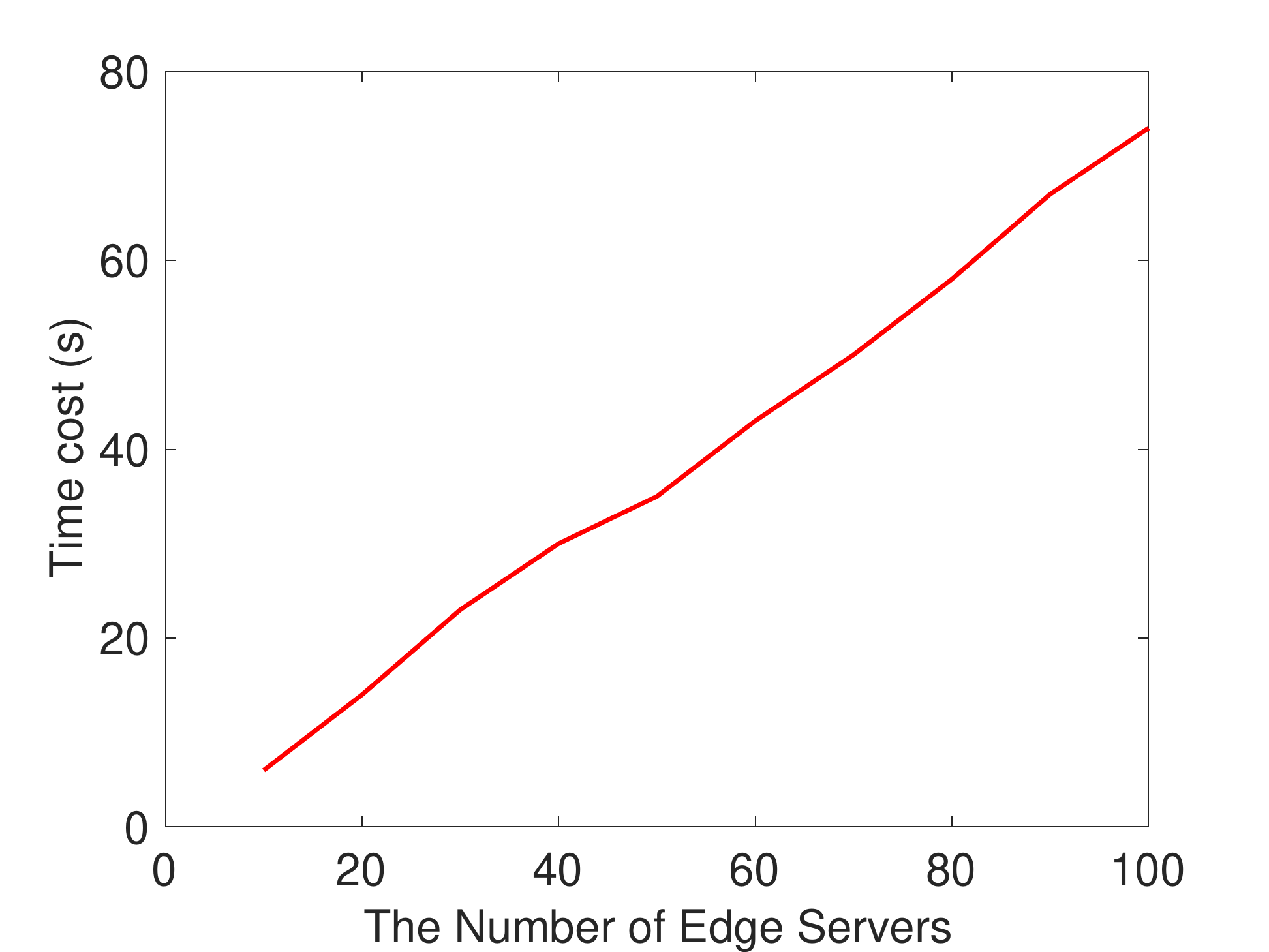}}
\subfigure{
\label{fig:clients}
\includegraphics[width=0.23\textwidth]{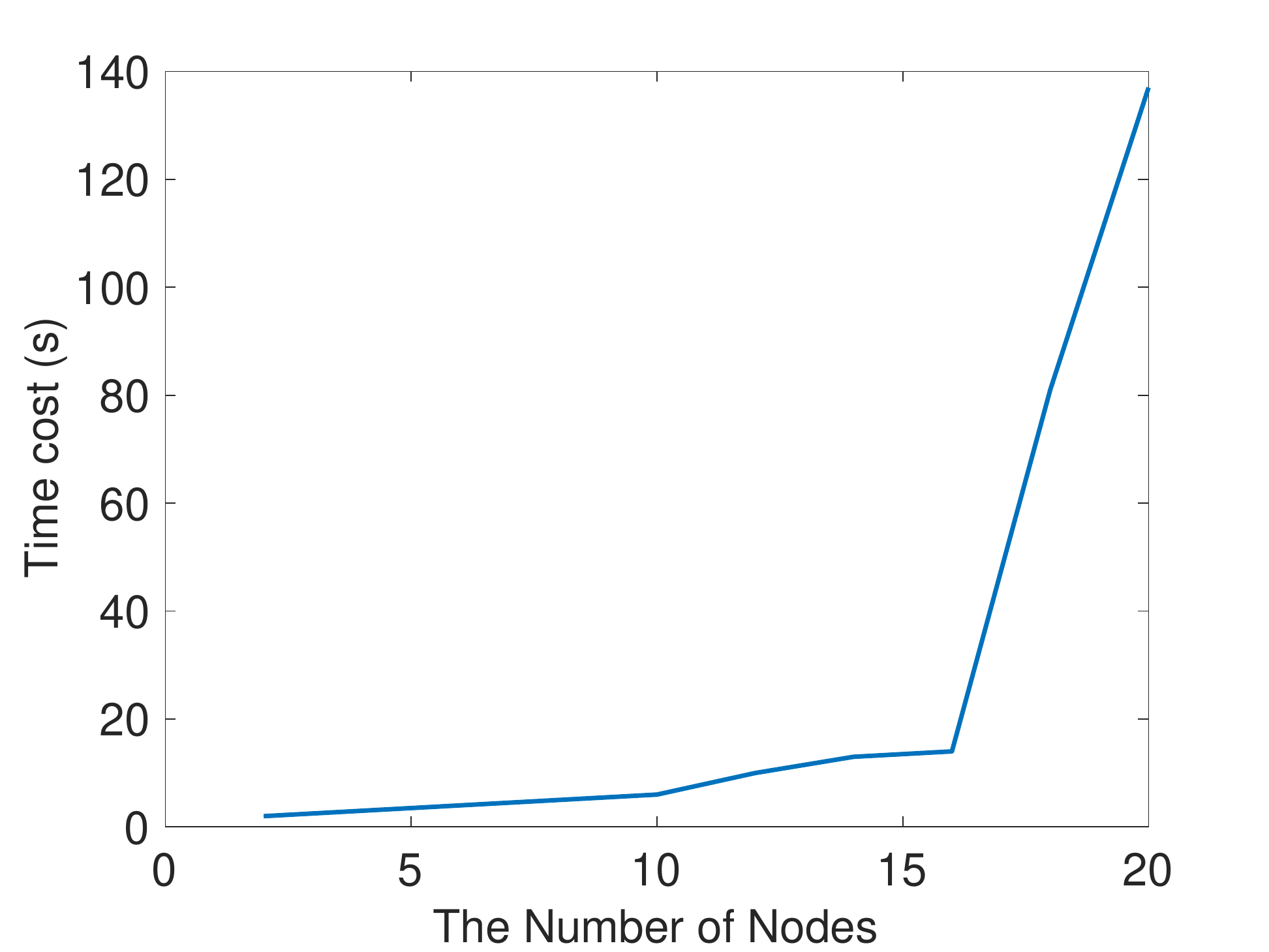}}
\caption{Impacts of the numbers of edge servers and blockchain nodes on the efficiency of the blockchain consensus.}
\label{fig:node}
\end{figure}

\subsection{ZD-based Payment Settlement}
To figure out whether the proposed ZD-based payment settlement strategy can function effectively, we compare it with other two classical strategies, named tit-for-tat (TFT) and win-stay-lose-shift (WSLS). With the TFT strategy, the leader will choose the action that the requestor adopted in the last round; with the WSLS strategy, the leader keeps performing an action if it brings a high payoff while changes to the other one if it results in a low payoff in the last round. 
Besides, we set the initial cooperation probability of the requestor as different values to further indicate the robustness of the proposed scheme. Specifically, we denote the requestor's initial cooperation probability as $q^0$ and study the evolution of the requestor's cooperation probability with respect to $q_1$ or $q_2$ according to the action of the leader in the current round. Main parameters in this scheme are payoff vectors of two player, which are set as $\mathbf{x} = (6,3,8,5)$ and $\mathbf{y} = (6,8,3,5)$ in the simulations.

From Fig. \ref{fig:ex3_probability}, one can see that the proposed ZD strategy can enable the cooperation probability of the requestor to gradually approach 1 and stay cooperative. While the other two classical strategies fail to achieve this job. Further, as presented in four subfigures, no matter how the initial cooperation probability of the requestor changes, the ZD-based scheme can always function successfully with respect to enforcing its full cooperation behavior, i.e., paying the full amount of incentive to the blockchain system. Thus, the effectiveness and robustness of our proposed scheme are validated.

\begin{figure}[htbp]
\centering
\subfigure[$q^0=0.10$.]{
\includegraphics[width=0.23\textwidth]{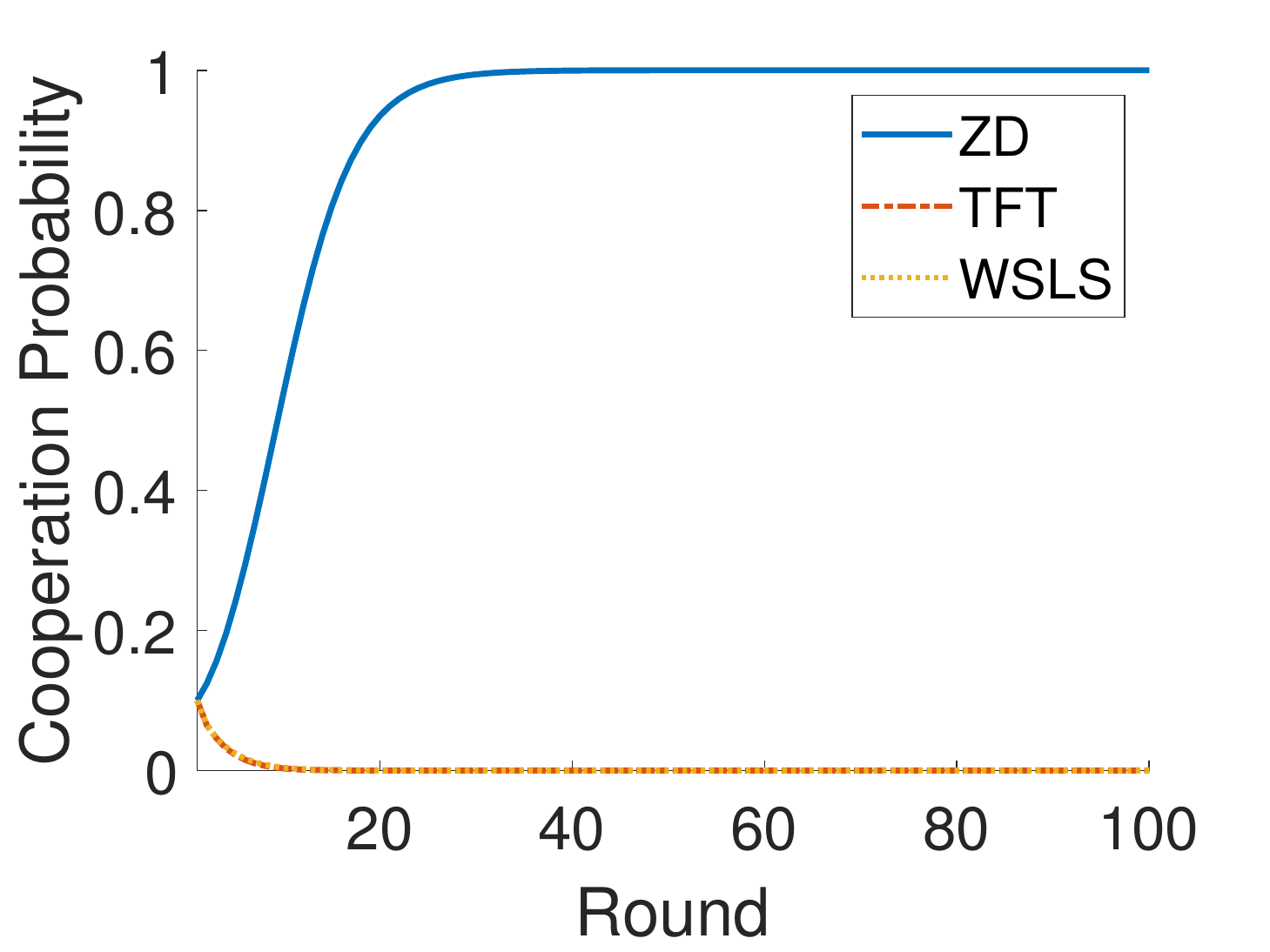}}
\subfigure[$q^0=0.40$.]{
\includegraphics[width=0.23\textwidth]{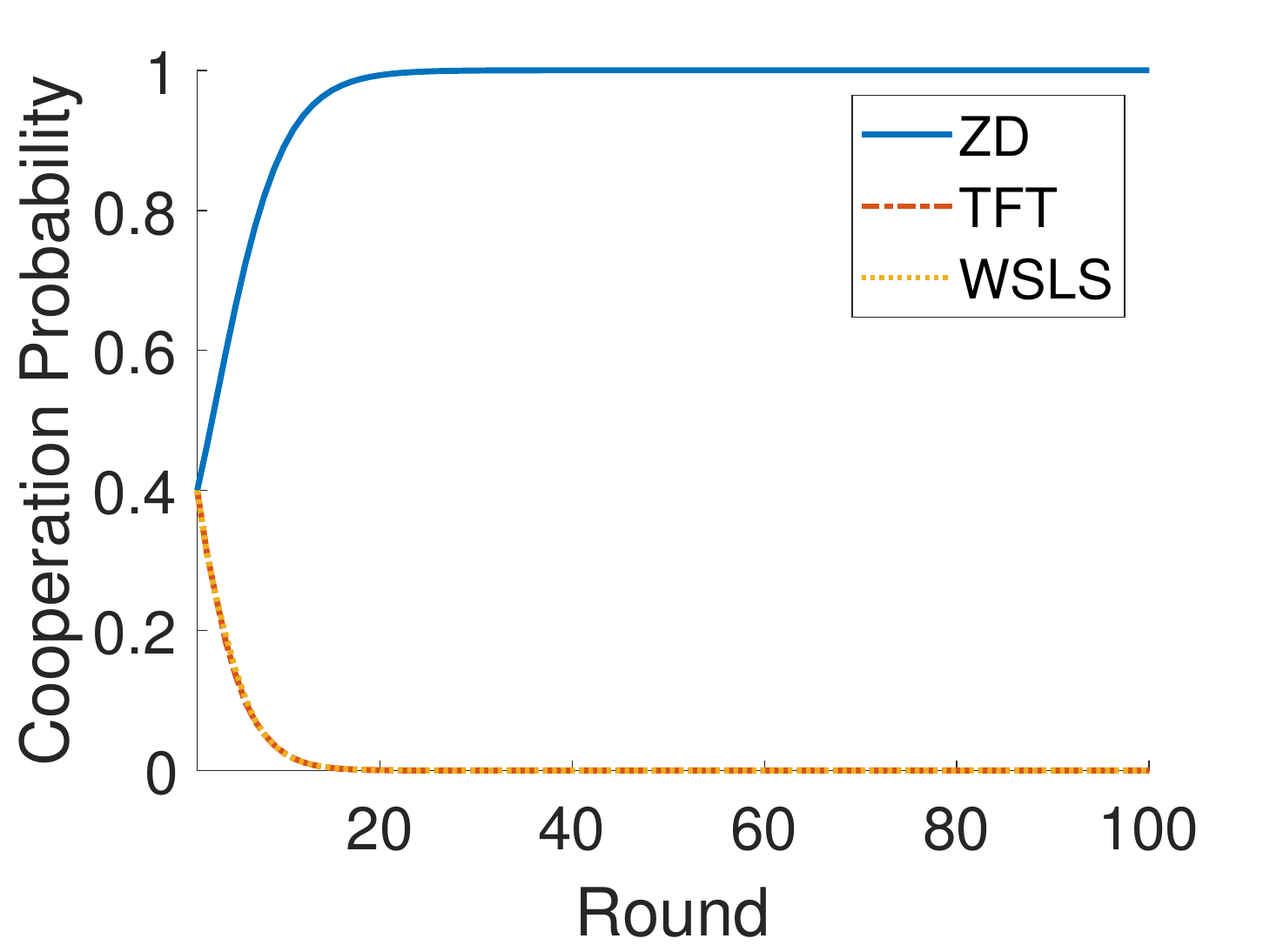}}
\subfigure[$q^0=0.60$.]{
\includegraphics[width=0.23\textwidth]{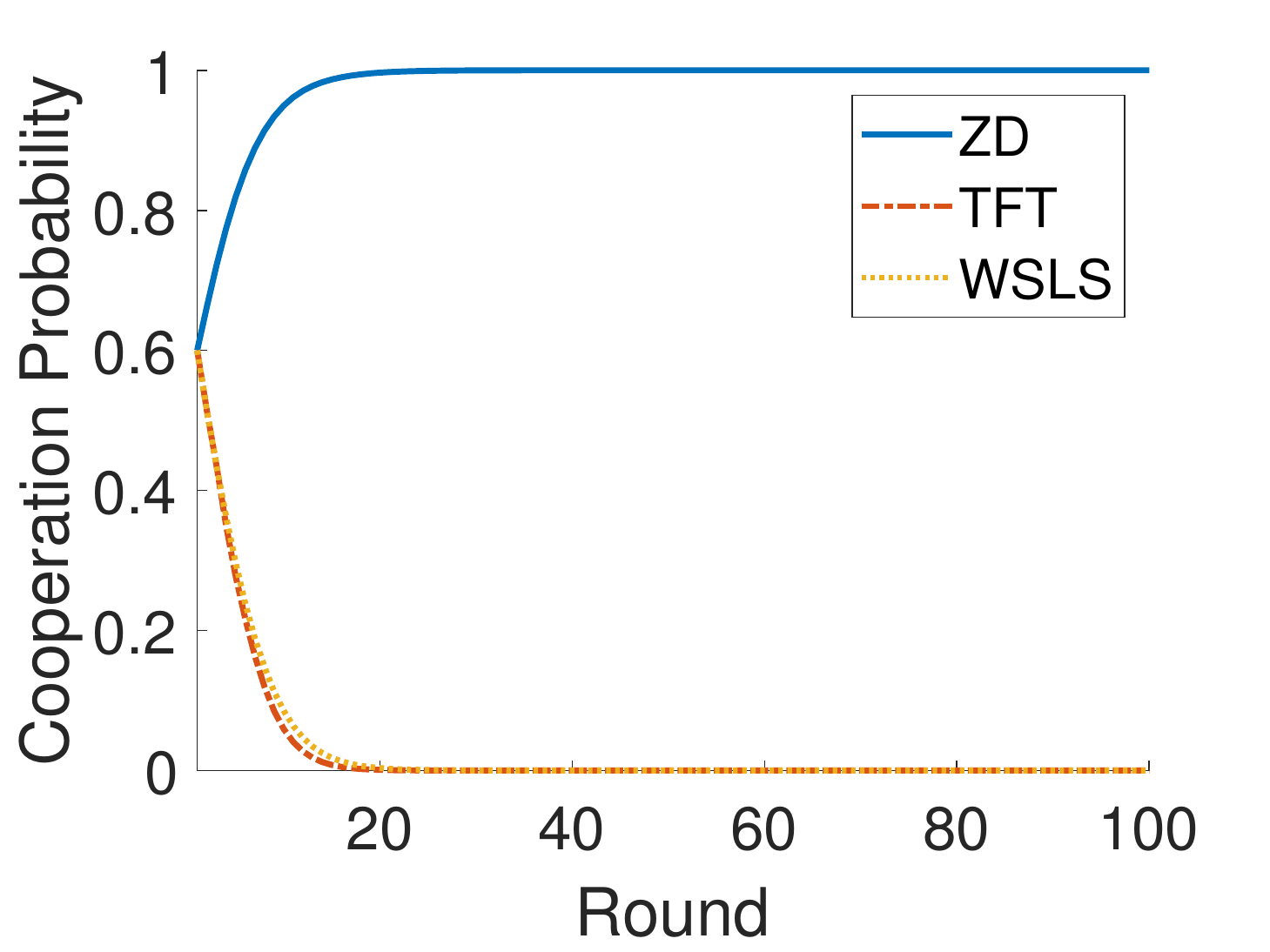}}
\subfigure[$q^0=0.90$.]{
\includegraphics[width=0.23\textwidth]{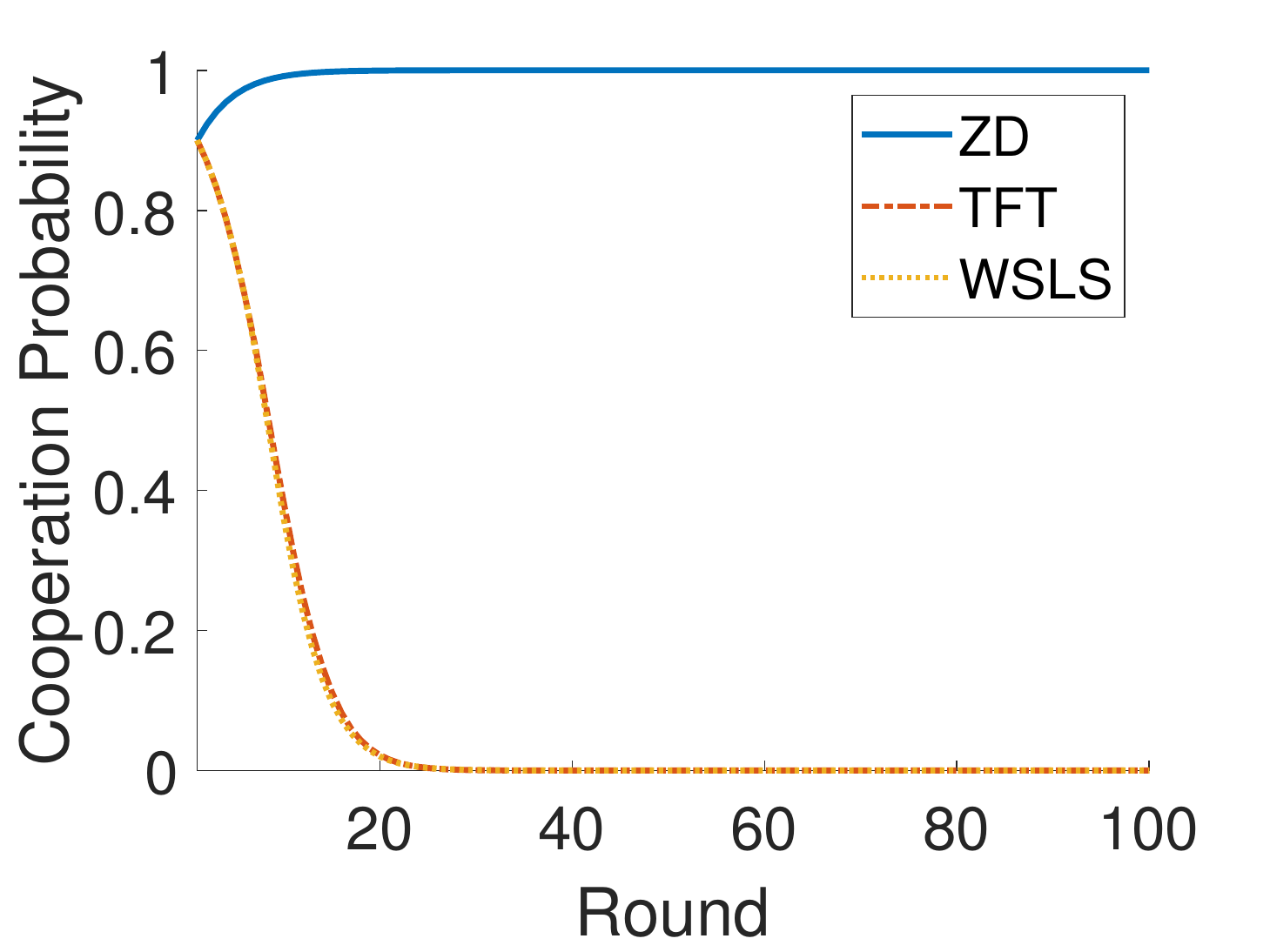}}
\caption{Cooperation probability of the requestor given different strategies of the leader.}
\label{fig:ex3_probability}
\end{figure}

Next, we explore the utilities of both the leader and requestor under the function of the proposed ZD-based scheme. As shown in Fig. \ref{fig:ex3_payoff}, we present the utilities at the stable state in the left bar graph and the evolution of utilities in the right-side line graphs. Similar to the previous experiment, we examine the results with different initial cooperation probabilities of the requestor. As can been seen from the left subfigure of Fig. \ref{fig:ex3_payoff}, 
at the final state with stable strategies of the requestor and the leader, both can obtain similar utility around 6 no matter how cooperative the requestor is at the beginning, which is exactly the utility at the mutual cooperation state. This implies that the proposed scheme is fair with respect to the final utility. From four line figures, one can observe that with different initial cooperation probabilities of the requestor, the dynamic evolution paths of both the leader and the requestor are generally similar with slight difference, where the higher the initial cooperation probability, the faster both sides to achieve the maximized utilities at the stable state. This is consistent with the fact that the more cooperative the requestor, the easier to drive its full cooperation at the stable state.

\begin{figure}[!htbp]
\centering
\subfigure{
\begin{minipage}{0.415\linewidth}
\includegraphics[width=1\textwidth]{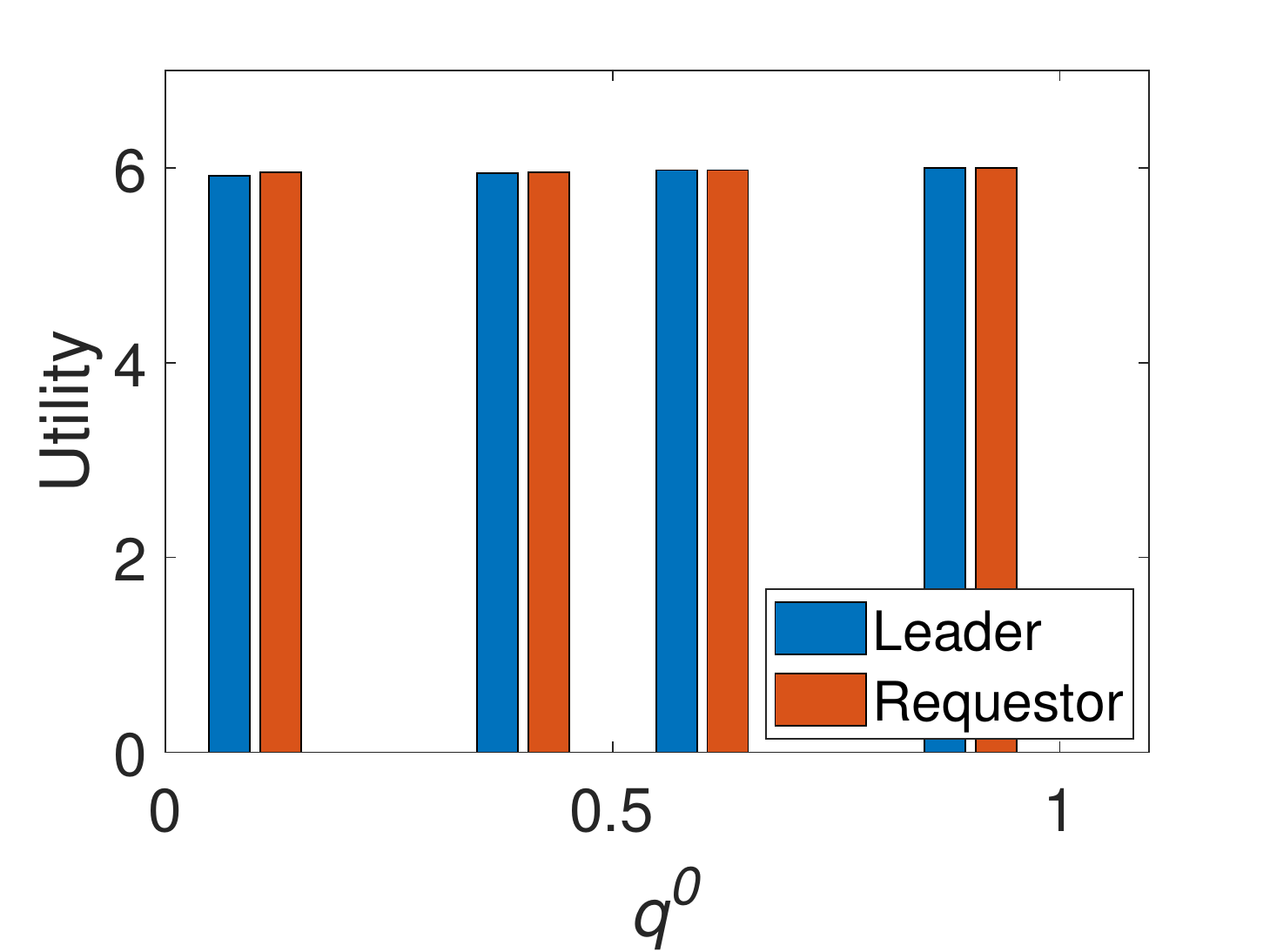}
\end{minipage}
}
\subfigure{
\begin{minipage}{0.25\linewidth}
\includegraphics[width=1\textwidth]{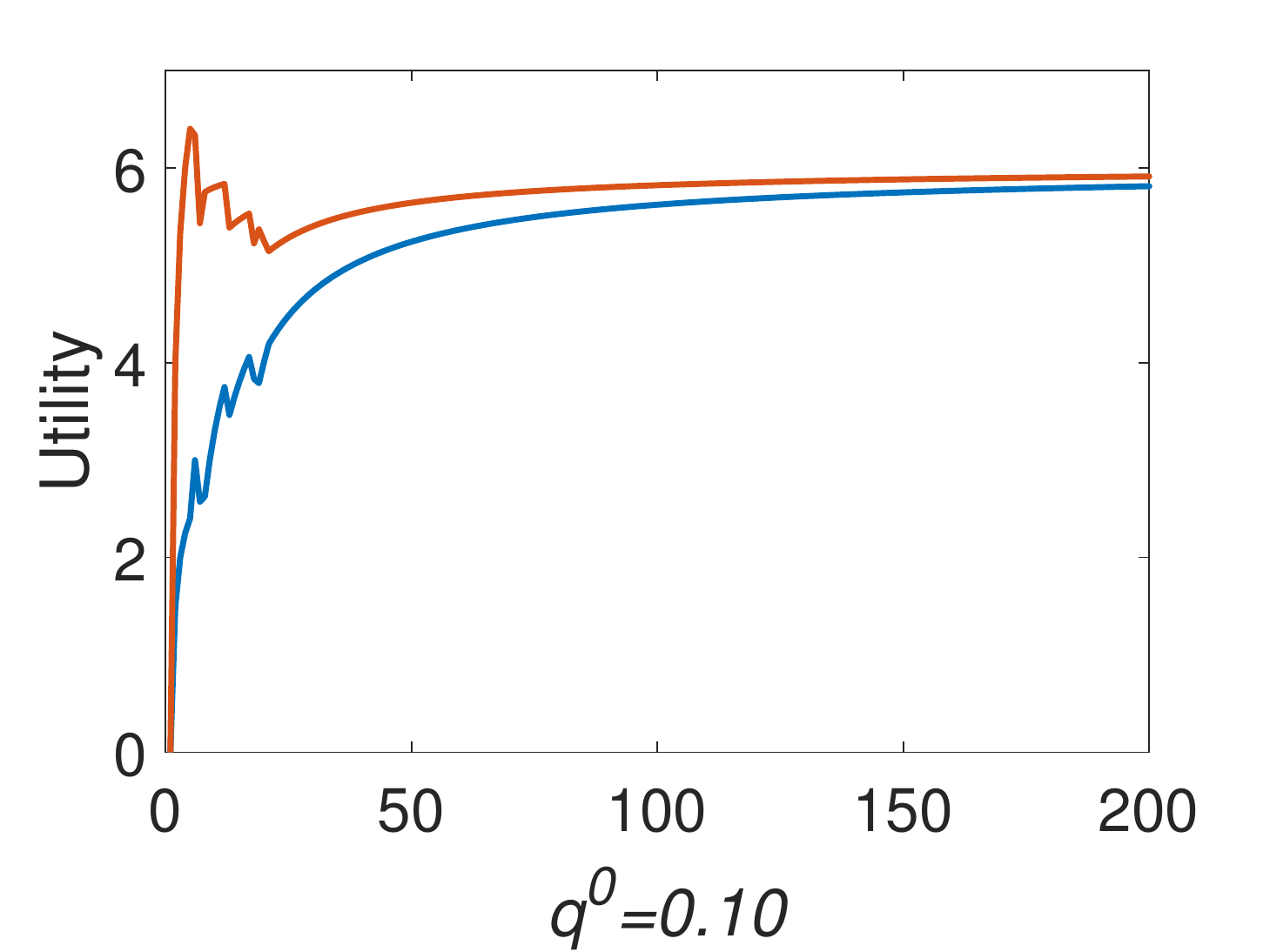}\\
\includegraphics[width=1\textwidth]{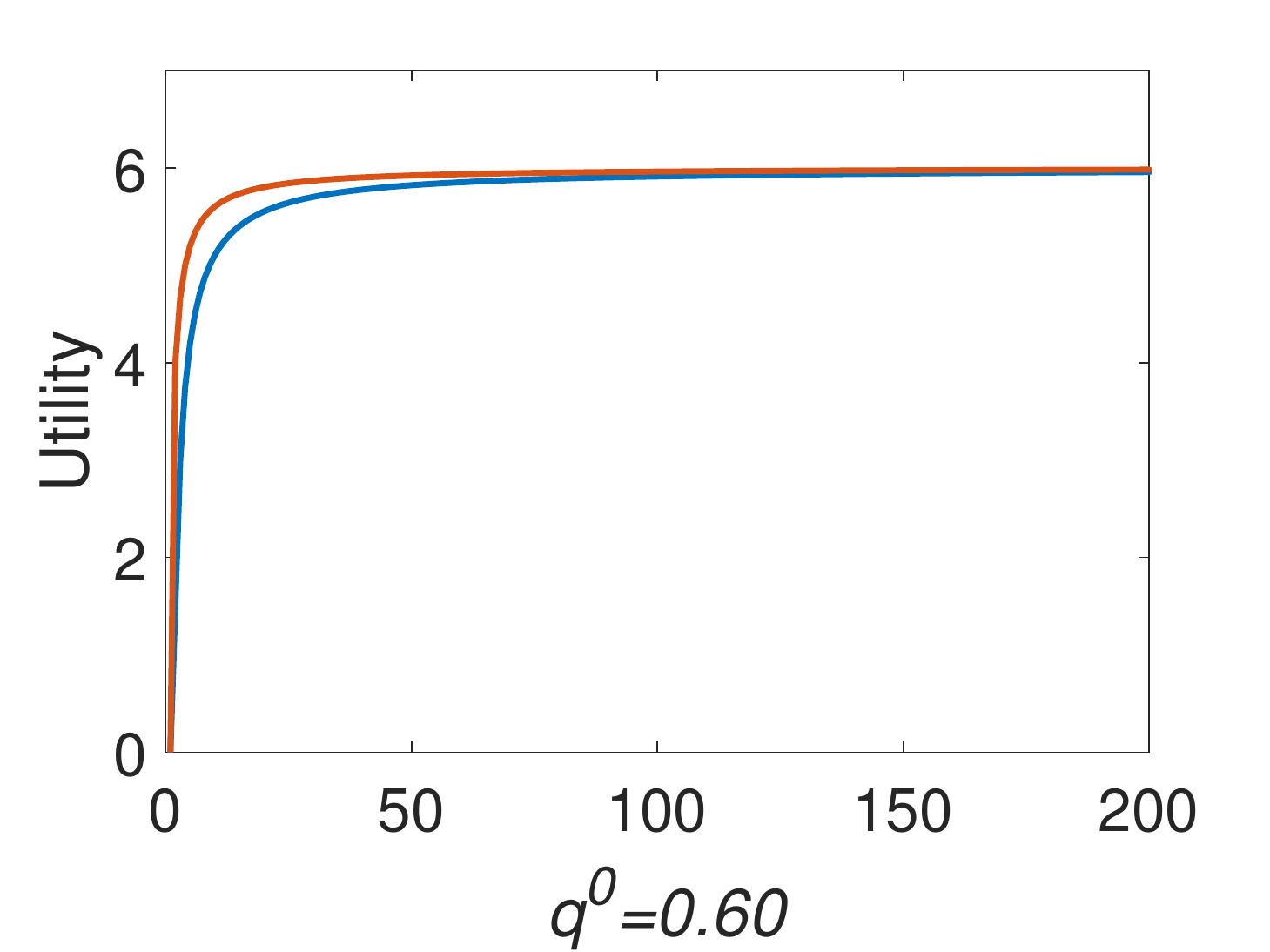}
\end{minipage}

\begin{minipage}{0.25\linewidth}
\includegraphics[width=1\textwidth]{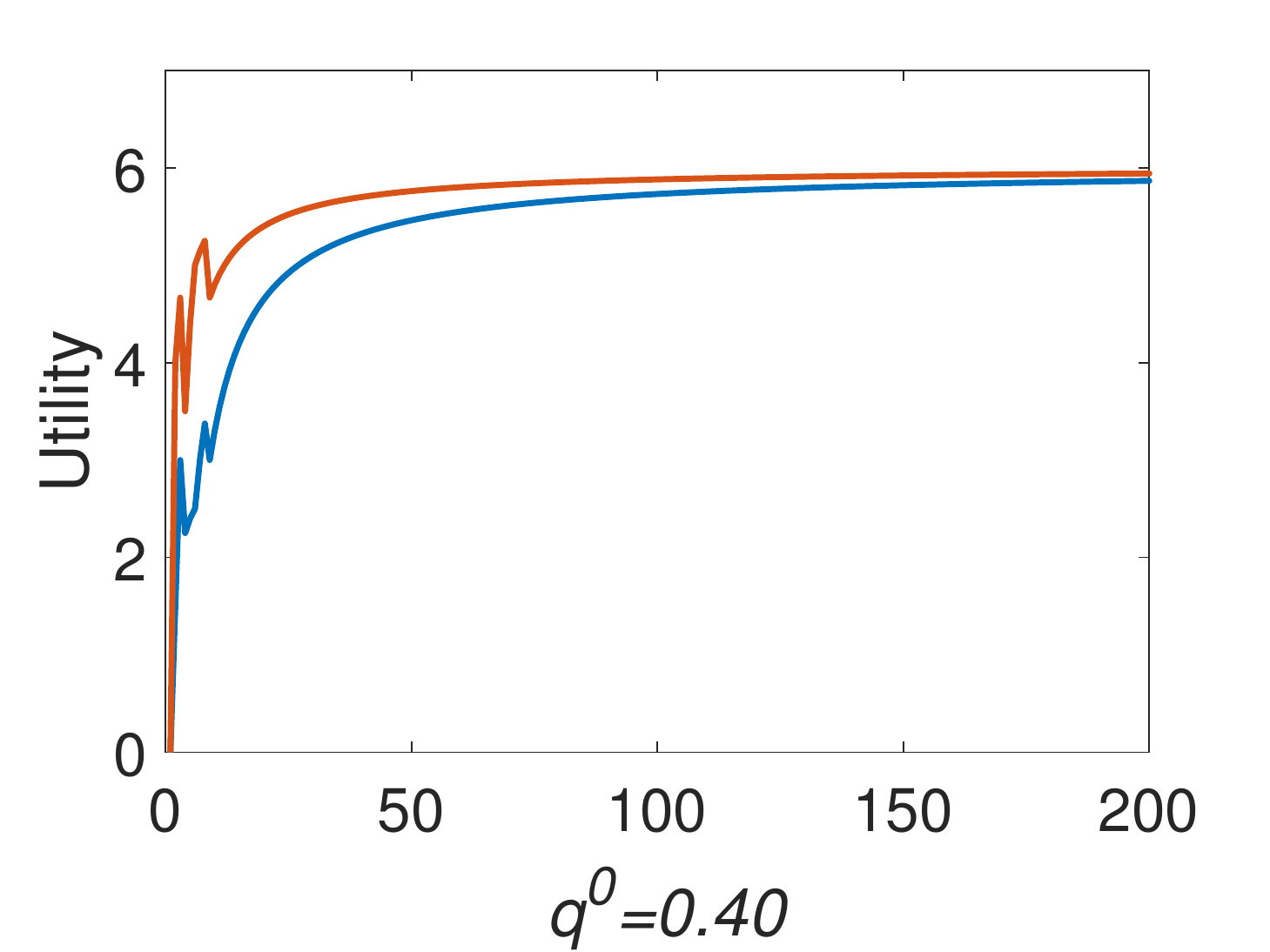}\\
\includegraphics[width=1\textwidth]{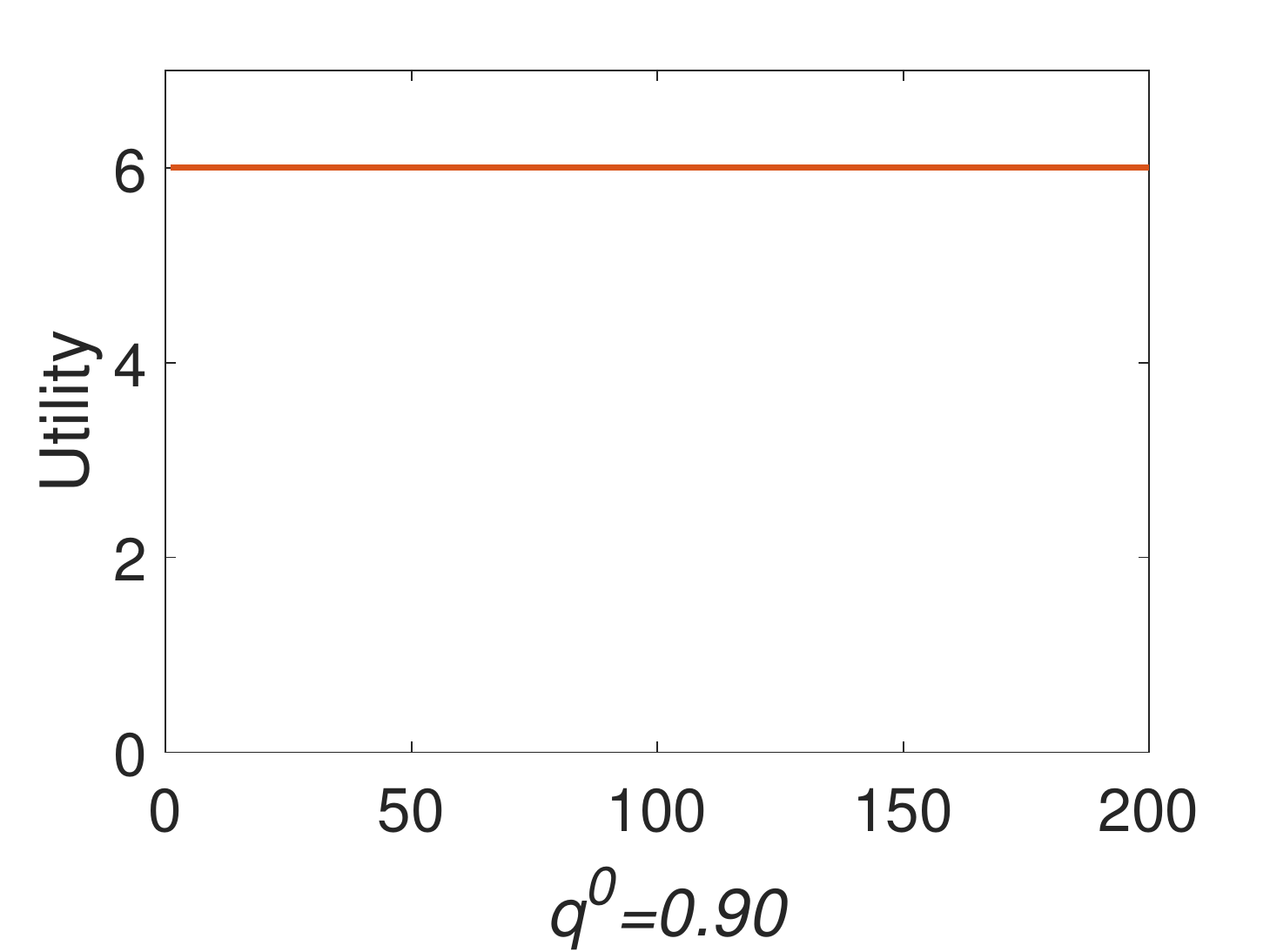}
\end{minipage}
}
\centering
\caption{Utilities of the leader and the requestor at the stable and dynamic states with the leader adopting the ZD strategy.}
\label{fig:ex3_payoff}
\end{figure}

\section{Conclusion and Future Work}\label{sec:conclusion}
Although MCS has been applied to various fields and brings significant benefits to the whole society, the current MCS paradigm has high requirements on the communication, storage and computing capabilities of requestors, limiting the widespread application of MCS to go a step further. To overcome this challenge, we propose a new MCS learning framework based on the consortium blockchain and the edge-participated FL, functioning among four major entities, i.e., requestors, blockchain, edge servers and mobile devices. Despite existing studies on blockchain for MCS and blockchain-based FL, they fail to lower the requirements on requestors' capabilities or cannot be directly applied to MCS. Thus, our proposed MCS learning framework fills the gaps with four main procedures, i.e., task publication, data sensing and submission, learning to return final results, and payment settlement and allocation. We design specific schemes in main steps to address challenges resulted from malicious edge servers, dishonest requestors, and even outside attacks.
Experimental results demonstrate the effectiveness of our designed schemes.

In our future work, we will make efforts to extensively study the employed blockchain in our proposed MCS learning framework to further improve the efficiency, scalability and reliability; besides, we will also investigate security challenges in blockchain-based distributed learning during the model training process, which may benefit not only our proposed MCS learning system but also other applications involving extensive data collection and analysis.



\bibliographystyle{IEEEtran}
\bibliography{reference}

\begin{IEEEbiography}[{\includegraphics[width=1in,height=1.25in,clip,keepaspectratio]{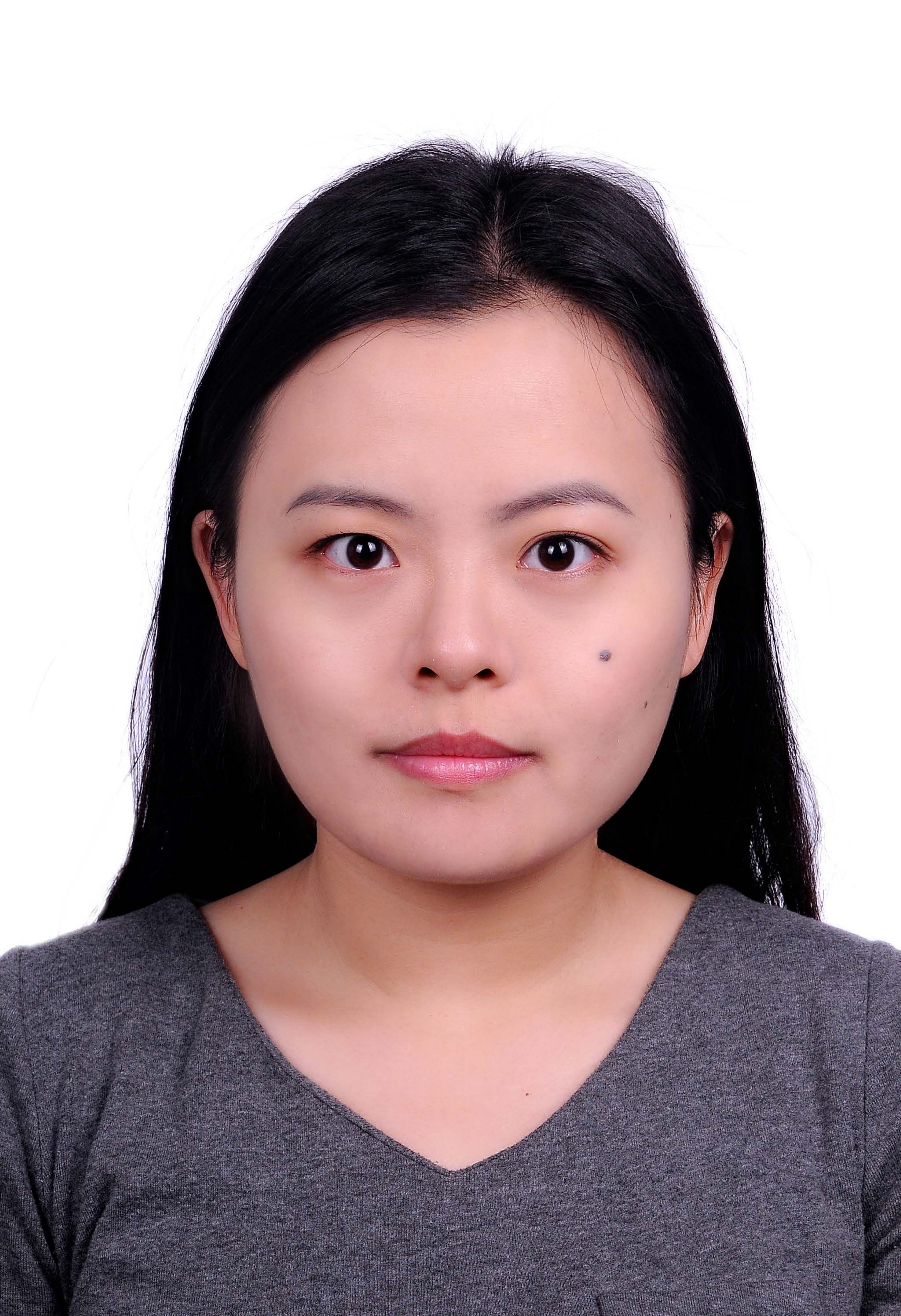}}]{Qin Hu} received her Ph.D. degree in Computer Science from the George Washington University in 2019. She is currently an Assistant Professor with the Department of Computer and Information Science, Indiana University-Purdue University Indianapolis (IUPUI). Her research interests include wireless and mobile security, edge computing, blockchain, and crowdsourcing/crowdsensing.
\end{IEEEbiography}

\begin{IEEEbiography}[{\includegraphics[width=1in,height=1.25in,clip,keepaspectratio]{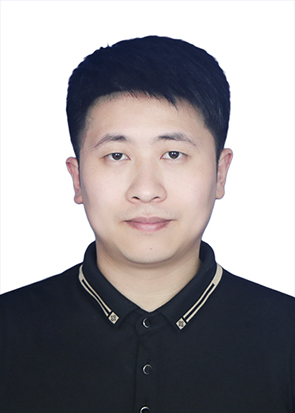}}]{Zhilin Wang} received his B.S. from Nanchang University in 2020. He is currently pursuing his Ph.D. degree of Computer and Information Science In Indiana University-Purdue University Indianapolis (IUPUI). His research interests include blockchain, federated learning, and Internet of Things (IoT).
\end{IEEEbiography}

\begin{IEEEbiography}[{\includegraphics[width=1in,height=1.25in,clip,keepaspectratio]{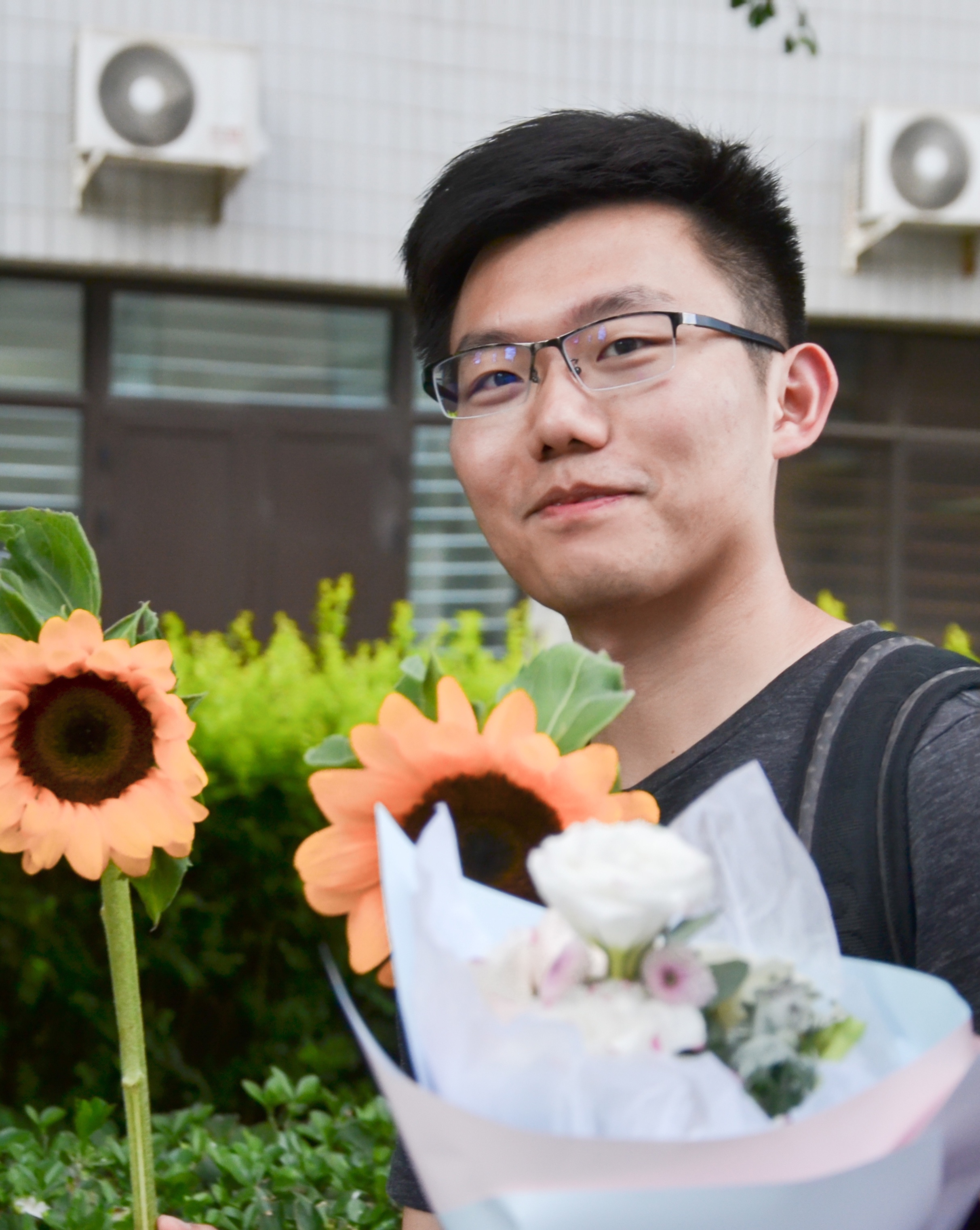}}]{Minghui Xu} received the BS degree in Physics from the Beijing Normal University, Beijing, China, in 2018, and the PhD degree in Computer Science from The George Washington University, Washington DC, USA, in 2021. He is currently an Assistant Professor in the School of Computer Science and Technology, Shandong University, China. His current research focuses on blockchain, distributed computing,  and quantum computing.
\end{IEEEbiography}

\begin{IEEEbiography}[{\includegraphics[width=1in,height=1.25in,clip,keepaspectratio]{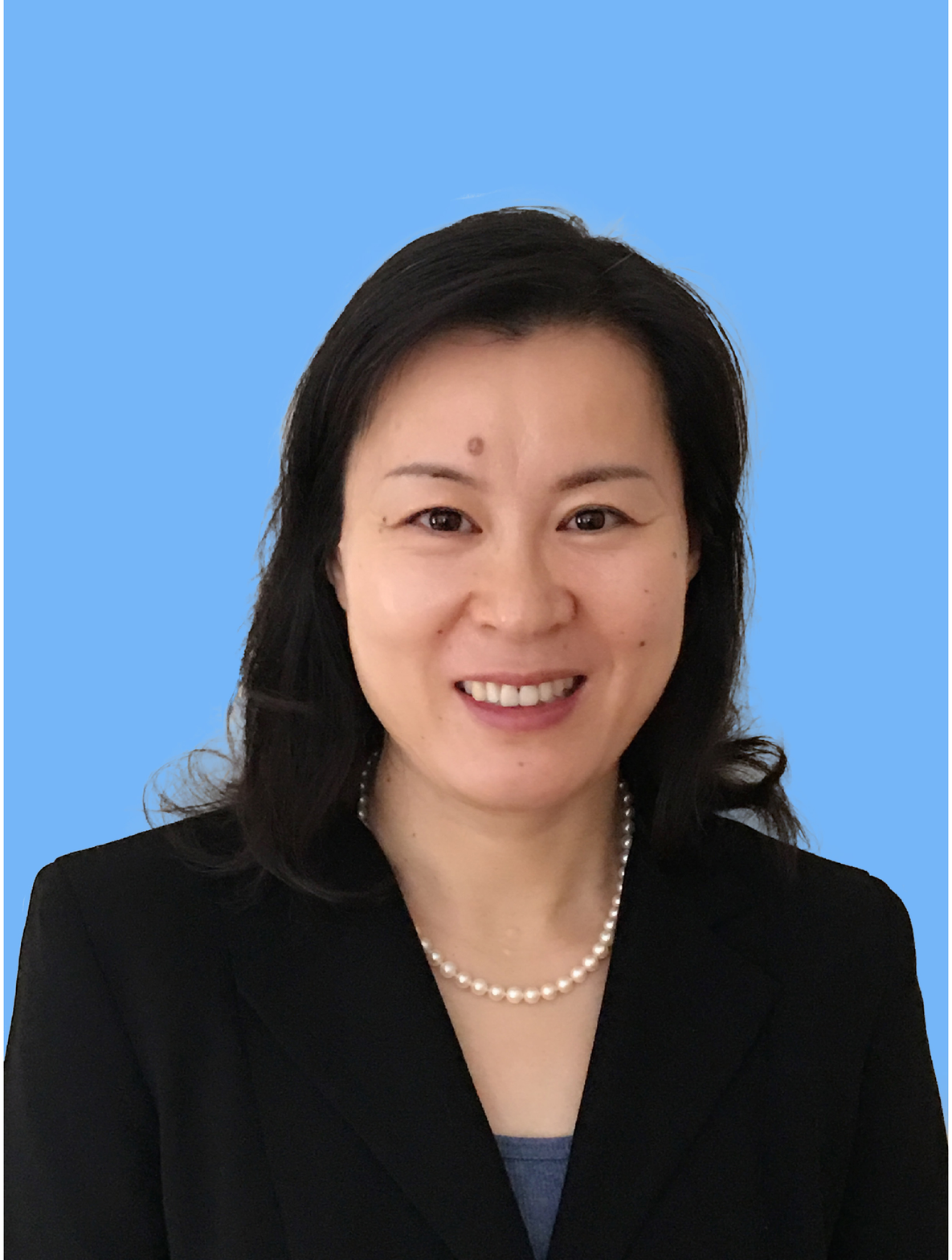}}]{Xiuzhen Cheng}
received her M.S. and Ph.D.
degrees in computer science from the University of
Minnesota Twin Cities in 2000 and 2002. She is currently a
professor in the School
of Computer Science and Technology, Shandong University, China. 
Her current research interests focus on privacy-aware
computing, wireless and mobile security, dynamic
spectrum access, mobile handset networking systems
(mobile health and safety), cognitive radio networks,
and algorithm design and analysis. She has served
on the Editorial Boards of several technical publications and the Technical Program Committees of various professional
conferences/workshops. She has also chaired several international conferences.
She worked as a program director for the U.S. National Science Foundation
(NSF) from April to October 2006 (full time), and from April 2008 to May
2010 (part time). She published more than 300 peer-reviewed papers.
\end{IEEEbiography}


%

\end{document}